\documentclass[10pt,authoryear]{article}
\usepackage{appendix}
\usepackage{setspace}
\usepackage[font=small,  margin=2cm]{caption}
\usepackage[tbtags]{amsmath}
\usepackage{amsthm,amssymb,amsfonts}
\usepackage{natbib}
\usepackage{multirow}
\usepackage{lscape}
\usepackage{subfigure}
\usepackage{makecell}
\usepackage{exscale}
\usepackage{booktabs}
\usepackage{array}
\usepackage{fullpage}
\usepackage{url}
\usepackage{algorithm}
\usepackage{algpseudocode}
\usepackage{bm}
\usepackage{smile}
\usepackage{mathtools}
\usepackage{wrapfig}
\usepackage{lipsum}
\usepackage{threeparttable}
\usepackage{mathrsfs}
\usepackage{relsize}
\usepackage{dsfont}
\usepackage{multirow}
\usepackage{apalike}
\usepackage[usenames,dvipsnames,svgnames,table]{xcolor}
\usepackage[colorlinks=true,linkcolor=blue,urlcolor=blue,citecolor=blue]{hyperref}

\newcommand{\bmin}{\mathop{\mathrm{min}}}
\newcommand{\bargmin}{\mathop{\mathrm{arg\ min}}}
\numberwithin{equation}{section}
\numberwithin{theorem}{section}
\numberwithin{corollary}{section}
\numberwithin{definition}{section}

\begin{document}

\title{\LARGE Large-dimensional Factor Analysis without Moment Constraints}

	\author{Yong He\thanks{ Institute for Financial Studies, Shandong University, Jinan, China; Email:{\tt heyong@sdu.edu.cn}.},~~Xinbing Kong\thanks{Nanjing Audit University, Nanjing, 211815, China; Email:{\tt xinbingkong@126.com}.},~~Long Yu\thanks{School of Management, Fudan University, Shanghai, China; Email:{\tt loyu@umich.edu}.},~~Xinsheng Zhang\thanks{School of Management, Fudan University, Shanghai, China; Email:{\tt xszhang@fudan.edu.cn}.}}	
	\date{}	
	\maketitle
Large-dimensional factor model has drawn much attention in the big-data era, in order to reduce the dimensionality and extract underlying features using a few latent common factors. Conventional methods for estimating the factor model typically requires finite fourth moment of the data, which ignores the effect of heavy-tailedness and thus may result in unrobust or even inconsistent estimation of the factor space and common components. In this paper, we propose to recover the factor space by performing principal component analysis to the spatial Kendall's tau matrix instead of the sample covariance matrix. In a second step, we estimate the factor scores by the ordinary least square (OLS) regression. Theoretically, we show that under the elliptical distribution framework the factor loadings and scores as well as the common components can be estimated consistently without any moment constraint. The convergence rates of the estimated factor loadings, scores and common components are provided. The finite sample performance of the proposed procedure is assessed through thorough simulations. An analysis of a financial data set of asset returns shows the superiority of the proposed method over the classical PCA method.

\vspace{2em}


\section{Introduction}
Factor model is a classical statistical model that serves as an important dimension reduction tool by characterizing the dependency structure of variables via a few latent factors. In the ``big-data era" where more and more variables are recorded and stored, large-dimensional approximate factor model is drawing growing attention as it
provides an effective way of summarizing information from large data sets. The large-dimensional approximate factor models are widely used in genomics, neuroscience, computer science and financial economics.  Theoretical analysis of large-dimensional approximate factor models has been studied by many researchers. Existing factor analysis procedures mainly fall into two categories: the principle component analysis (PCA) approach and the maximum likelihood estimation (MLE) method. The PCA-based method is easy to implement and  provides consistent estimators for the factors and factor loadings when both the cross-section $p$ and time dimension $n$ are large. Representative works include, but not limited to, \cite{Bai2002Determining,stock2002forecast,Stock2002Macroeconomic,Bai2003Inferential,onatski2009testing,Ahn2013Eigenvalue,
fan2013large,Trapani2018A}. It turns out that the PCA approach is equivalent to the least square optimization. The MLE-based method is more efficient than the PCA-based approach but is also computationally more suffering. Representative works, to name a few, are \cite{Bai2012Statistical,Bai2014Theory,Bai2016Maximum}.

However, the aforementioned works all assume that the fourth moments (or even higher moments) of factors and idiosyncratic errors are bounded such that the least-squares regression, or maximum likelihood estimation can be applied. This assumption is really an idealization of the complex random real world. Heavy-tailed data are often encountered in scientific fields such as financial engineering and biomedical imaging. In finance,  \cite{Fama1963Mandelbrot} discussed the power law behavior of asset returns. \cite{Cont2001Empirical} provided extensive empirical evidence of heavy-tailedness in financial returns. \cite{jing2012modeling} and  \cite{kong2015testing} even suggested to model the log price dynamics of an asset by pure jump processes without any moment conditions. Thus, it is imperative to develop estimation procedures that are robust to heavy-tailedness for large-dimensional factor models.

\begin{figure}[h]
  \centering
  \begin{minipage}[!t]{0.48\linewidth}
    \includegraphics[width=1\textwidth]{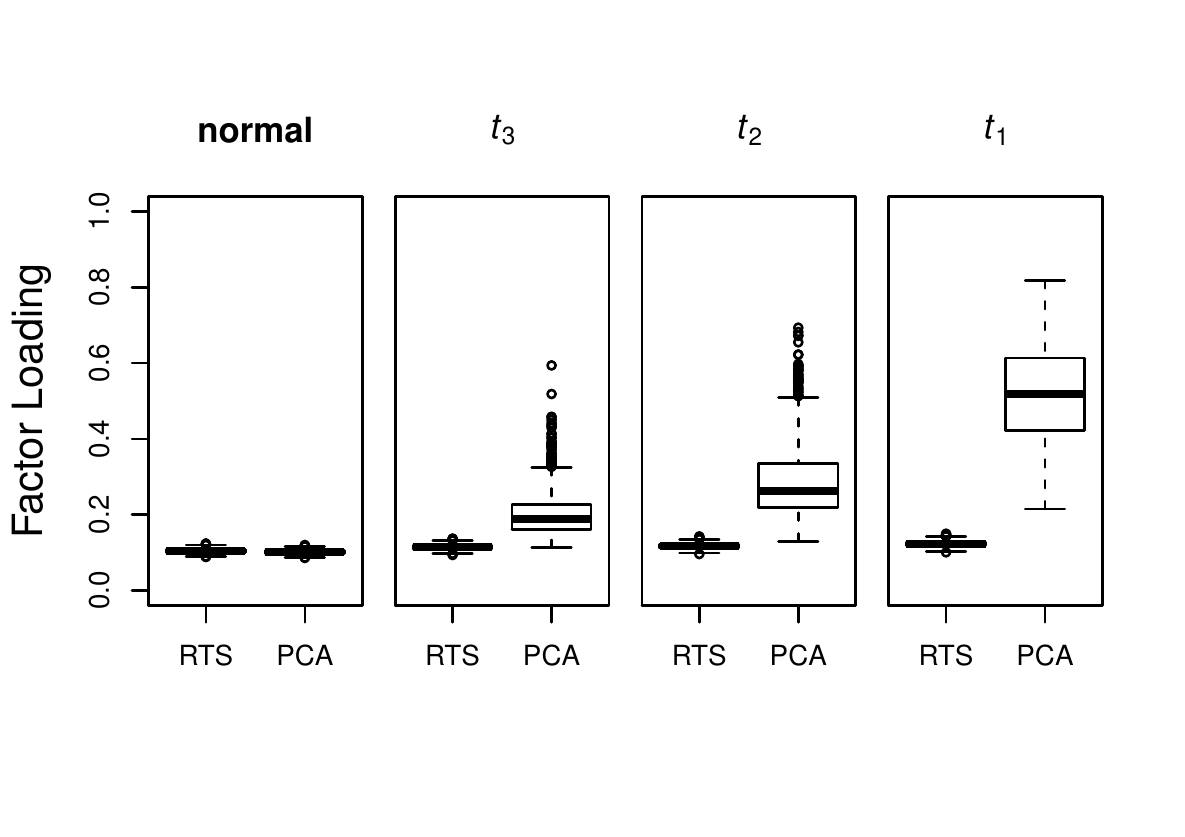}\\
  \end{minipage}
  \begin{minipage}[!t]{0.48\linewidth}
    \includegraphics[width=1\textwidth]{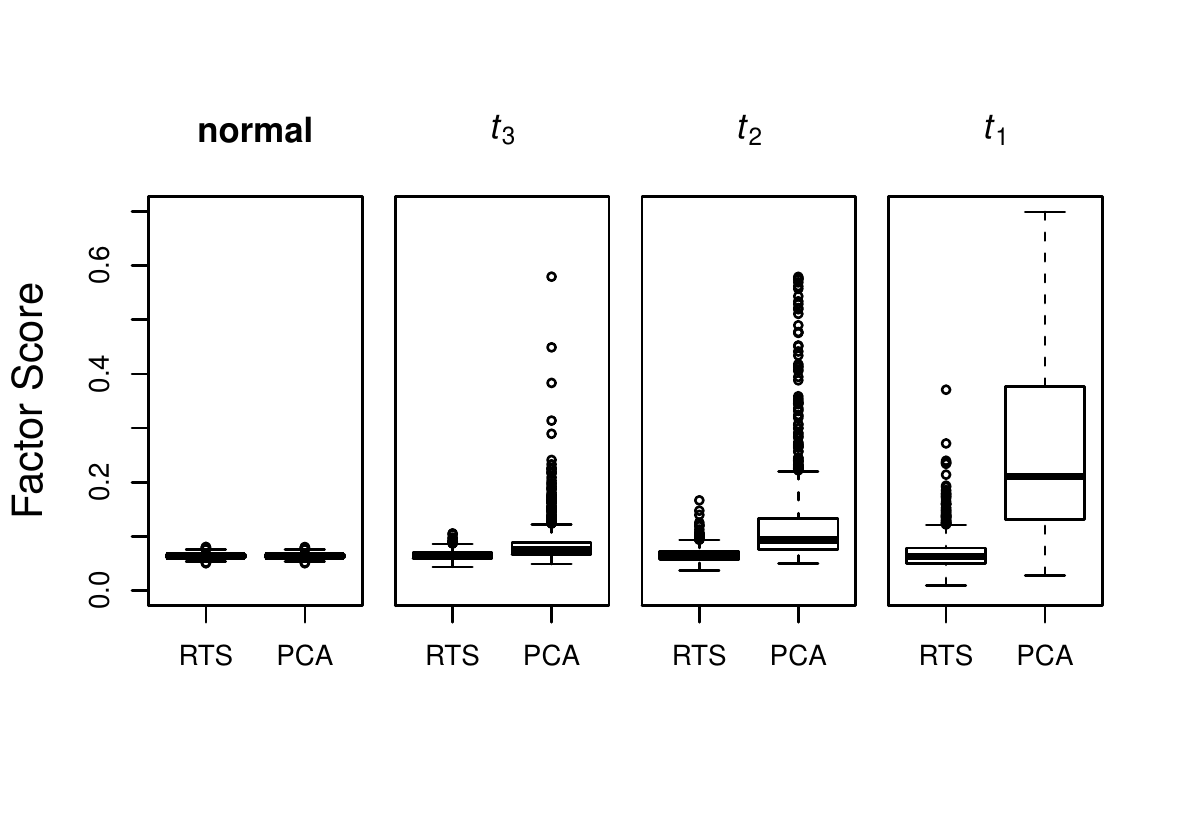}\\
  \end{minipage}
 \caption{Boxplots of the estimation errors of the estimated factor loadings and scores by RTS and PCA methods under different data generating distributions---normal, $t_3$, $t_2$ and $t_1$. $p=250,n=100$.}\label{fig:1}
 \end{figure}

As an illustration, we check the sensitivity of the PCA (or Least Square Optimization) method to the heavy-tailedness of the factor and idiosyncratic errors with a synthetic data set. We generate the factors and idiosyncratic errors from joint normal, $t_3$, $t_2$ and $t_1$ distributions that will be described in detail in Section \ref{ss}. Figure \ref{fig:1} depicts the boxplots of the factor loading and factor score estimation errors based on 1000 replications. We observe that the PCA results in bigger biases and higher dispersions as the distribution tails become heavier. This is also consistent with the fourth moment condition imposed in existing papers.

In this article, we propose a robust two step (RTS) procedure to estimate the factor loadings, scores and common components without any moment constraint under the framework of elliptical distributions (FED). The FED assumes that the factors and the idiosyncratic errors jointly follow an elliptical distribution, which covers a large class of heavy-tailed distributions such as $t$-distribution. The FED is drawing growing attention as an important tool to simultaneously simplify the structure and capture the heavy-tailedness of the data. For example, \cite{fan2018} considered large-scale covariance estimation through the FED; and  \cite{yu2018robust} proposed robust estimator of the factor number of a large-dimensional factor model under the FED condition. In the first step, we recover the factor space, spanned by the columns of the factor loadings, by performing PCA to the estimated spatial Kendall's tau matrix rather than the sample covariance matrix. The spatial Kendall's tau matrix shares the same eigenspace with the scatter matrix of the elliptically distributed data vectors and the scatter matrix serves as a measure of the cross-sectional dependence of the data. Since the scatter matrix has nothing to do with the moment of the data, the resulting estimated factor space from RTS is consistent to the true factor space without any moment requirement. In a second step, we estimate the factor score by a cross-sectional least square regression based on the estimated factor loadings in the first step. Due to the polarization of the elliptical distribution, the estimated factor scores are consistent up to some orthogonal transformation without any moment restriction. To the best of our knowledge, this is the first work that can estimate the factor loadings and scores (up to some orthogonal transformations) and common components for heavy-tailed data without any moment requirement on the factors and idiosyncratic errors under the FED condition. Now, let us come back to the example mentioned earlier in Figure \ref{fig:1}, in which we also presented the results using the RTS method. Figure \ref{fig:1} demonstrates that the estimated factor loadings and scores are not much sensitive to the heavy-tailedness of the factors and idiosyncratic errors. All the RTS estimates are substantially more accurate than the PCA estimates for the three $t$-distribution settings.

The critical tool in the current paper, spatial Kendall's tau matrix, is first introduced in \cite{Choi1998A}, also named as multivariate Kendall's tau matrix in the literature. Its applications on covariance matrix estimation and principal component analysis in low dimensions can be found in but not limited to \cite{Marden1999Some,Visuri2000Sign} and \cite{croux2002sign}. For high-dimensional settings, \cite{han2018eca} studied the eigen-analysis of spatial Kendall's tau matrix for elliptical  distributions.  Similar rank-based methods in high dimensions are also discussed in \cite{han2012semiparametric} and \cite{Han2014Scale}. The mentioned literatures provide a sound framework of constructing Bernstein's type concentration inequalities for high-dimensional matrix-form U-statistics, but  none focused on the factor models. As a result, the existing convergence rates for robust principal component analysis are not optimal if the observed vectors contain some low-rank factor structures. To the best of our knowledge, \cite{fan2018} is the the first to consider factor structures with spatial Kendall's tau matrix in high dimensions, which proposed the elliptical factor model  exactly the same as the model considered in the current paper. However, the motivations of the two works are quite different. \cite{fan2018} focused on robust covariance estimation while our work focused on robust estimation for the factor loadings and scores.

The contributions of the current paper lie in the following aspects. Firstly, it's the first to estimate the factor scores and loadings using spatial Kendall's tau matrix in high-dimensional settings. The proposed method is computationally efficient and easy to implement.  Secondly, our theoretical analysis shows that the proposed robust estimators are consistent without any moment constraints on the underlying distributions, which is the first in the literature and  makes the method applicable to analyzing heavy-tailed datasets such as financial returns and macroeconomic indicators. Thirdly, the convergence rates of the RTS estimators are the same as those in \cite{Bai2003Inferential}, which makes RTS a safe replacement of the conventional PCA approach. We overcome two major challenges in the technical proofs: 1) the summing terms in the sample spatial Kendall's tau matrix are dependent, which makes the typical Bernstein's inequities unapplicable; 2) the sample spatial Kendall's tau matrix is a  nonlinear function of the observed vectors, which in essence makes the theoretical analysis more challenging.

We introduce the notations adopted throughout the paper. For any vector $\bmu=(\mu_1,\ldots,\mu_p)^\top \in \RR^p$, let $\|\bmu\|_2=(\sum_{i=1}^p\mu_i^2)^{1/2}$, $\|\bmu\|_\infty=\max_i|\mu_i|$. For a real number $a$, denote  $[a]$ as the largest integer smaller than or equal to $a$. Let $I(\cdot)$ be the indicator function. Let ${\rm diag}(a_1,\ldots,a_p)$ be a $p\times p$ diagonal matrix, whose diagonal entries are $a_1\ldots,a_p$.   For a matrix $\Ab$, let $\mathrm{A}_{ij}$ (or $\mathrm{A}_{i,j}$) be the $ij$ entry of $\Ab$, $\Ab^\top$ the transpose of $\Ab$, ${\rm Tr}(\Ab)$ the trace of $\Ab$, $\text{rank}(\Ab)$ the rank of $\Ab$ and $\text{diag}(\Ab)$ a vector composed of the diagonal elements of $\Ab$. Denote $\lambda_j(\Ab)$ as the $j$-th largest eigenvalue of a nonnegative definitive matrix $\Ab$, and let $\|\Ab\|$ be the spectral norm of matrix $\Ab$ and $\|\Ab\|_F$ be the Frobenius norm of $\Ab$. For two series of random variables, $X_n$ and $Y_n$, $X_n\asymp Y_n$ means $X_n=O_p(Y_n)$ and $Y_n=O_p(X_n)$. For two random variables (vectors) $\bX$ and $\bY$, $\bX\stackrel{d}{=}\bY$ means the distributions of $\bX$ and $\bY$ are the same. Let $\one$ be a vector with all elements 1.  The constants $c,C_1,C_2$ in different lines can be nonidentical.

The rest of the paper proceeds as follows. In Section 2, we introduce the setup assumptions and multivariate Kendall's tau matrix. Estimators of the factor loadings, scores and common components are also provided. In Section 3, we establish the consistency including the convergence rate for the estimated factor loadings, scores and common components. Section 4 is devoted to a thorough numerical study. A real financial data set of asset returns is analyzed in Section 5. We discuss the possible future research directions and conclude the article in Section 6.  The  proofs of the main theorems  are collected in the Appendix and additional details are put in the supplement.

\section{Methodology}


\subsection{Elliptical distribution and spatial Kendall's tau matrix}\label{sec:2.1}

Consider the large-dimensional factor model for a large panel data set $\{y_{it}\}_{i\leq p,t\leq n}$,
\begin{equation}\label{EFM}
  y_{it}=\bl_i^\top \bbf_t+\epsilon_{it}, \hspace{0.5 em}  i\leq p, \ \   t\leq n, \hspace{0.5em} \text{or in vector form,} \hspace{0.5em} \by_t=\Lb \bbf_t+\bepsilon_t,
\end{equation}
where $\by_t=(y_{1t},\ldots,y_{pt})^\top$, $\bbf_t\in \RR^m$ are the unobserved factors, $\Lb=(\bl_1,\ldots,\bl_p)^\top$ is the factor loading matrix, and $\bepsilon_t=(\epsilon_{1t},\ldots,\epsilon_{pt})^\top$ represents the idiosyncratic errors. The term $C_{it}=\bl_i^\top\bbf_t$ is referred to as the common component of $y_{it}$. For the large-dimensional approximate factor model introduced in \cite{Chamberlain1983Arbitrage}, $\bepsilon_t$ is assumed to be cross-sectionally weakly dependent.

As mentioned in the introduction and precisely stated in Assumption A below, we assume that $(\bbf_t^\top,\bepsilon_t^\top)^\top$ is a series of temporally independent and identically distributed random vectors generated from an elliptical distribution. For a random vector $\bZ=(Z_1,\ldots,Z_p)^\top$ following an elliptical distribution, denoted by $\bZ\sim {ED}(\bmu,\bSigma,\zeta)$, we mean that
$$
\bZ\overset{d}{=}\bmu+\zeta\Ab\bU,
$$
where $\bmu\in \RR^p$, $\bU$ is a random vector uniformly distributed on the unit sphere $S^{q-1}$
in $\RR^q$, $\zeta\geq 0$ is a scalar random variable independent of $\bU$, $\Ab\in\RR^{p\times q}$ is
a deterministic matrix satisfying $\Ab\Ab^\top=\bSigma$ with $\bSigma$ called scatter matrix whose rank is $q$. Let the scatter matrices of $\bbf_t$ and $\bepsilon_t$ be $\bSigma_f$ and $\bSigma_{\epsilon}$, respectively. If $\bSigma_{\epsilon}$ is sparse, the model (\ref{EFM}) is indeed an approximate factor model including the strict factor model, in which $\bSigma_{\epsilon}$ is diagonal, as a special case. Another equivalent characterization of the elliptical distribution is by its characteristic function,
which has the form $\exp(i\bt^\top \bmu) \psi(\bt^\top\bSigma \bt)$, where $\psi(\cdot)$ is a properly defined characteristic function
and $i=\sqrt{-1}$.

The factor loadings and scores, $\Lb$ and $\bbf_t$, are not separately identifiable as they are unobservable. For an arbitrary $m\times m$ invertible matrix $\Hb$, one can always have $\Lb^*=\Lb\Hb$ and $\bbf^*_t=\Hb^{-1}\bbf_t$ such that $\Lb^*\bbf^*_t=\Lb\bbf_t$. For reason of identifiability, we impose the following constraints:
$$
\bSigma_{f}=\Ib_m \hspace{1em} \text{and} \hspace{1em} \big\|\text{diag}(\bSigma_\epsilon)\big\|_{\infty}=1.
$$
This identification condition is also used in \cite{Han2014Scale} and \cite{yu2018robust}, and it is not unique and one may refer to \cite{Bai2012Statistical} for more detailed discussion on identification issues.

It is worthy of pointing out that elliptical distributions have some nice properties as Gaussian distributions, e.g., the marginal distributions,
conditional distributions and distributions of linear combinations of elliptical vectors are also elliptical. Thus, for the factor model (\ref{EFM}) under the FED condition, the scatter matrix of $\by_t$, $\bSigma_{y}$, is composed of a low-rank part $\Lb\Lb^\top$ and a sparse part $\bSigma_{\epsilon}$, i.e, $\bSigma_{y}=\Lb\Lb^\top+\bSigma_{\epsilon}$.
For Gaussian distribution, $\bSigma_y$ is simply the population covariance matrix of $\by_t$. For non-Gaussian distributions, especially distributions with infinite variances, the scatter matrix is still a measure of the dispersion of a random vector. So, naturally the eigenspace of the scatter matrix $\bSigma_y$ sheds light into the recovery of the factor space, but this can not be reached by performing PCA to the sample covariance matrix because the covariance is meaningless for pairs of random variables with infinite variances. To tackle with this difficulty, we introduce the population spatial Kendall's tau matrix. Let $\bX\sim ED(\bmu,\bSigma,\zeta)$ and $\tilde{\bX}$ be an independent copy of $\bX$. The population spatial Kendall's tau matrix is defined as
$$
\Kb={\rm E}\left\{\frac{(\bX-\tilde{\bX})(\bX-\tilde{\bX})^\top}{\|\bX-\tilde{\bX}\|_2^2}\right\}.
$$
$\Kb$ can be estimated by a second-order U-statistic. Specifically, assume $\{\bX_1,\ldots,\bX_n\}$ is a series of $n$ independent data points following the distribution $\bX\sim ED(\bmu,\bSigma,\zeta)$. The sample version spatial Kendal's tau matrix is
$$
\hat{\Kb}=\frac{2}{n(n-1)}\sum_{t<t^\prime}\frac{(\bX_t-{\bX_{t^\prime}})(\bX_t-{\bX_{t^\prime}})^\top}{\|\bX_t-{\bX_{t^\prime}}\|_2^2}.
$$
The spatial Kendall' tau matrix was first introduced in \cite{Choi1998A} and has been used for covariance matrix estimation in \cite{Visuri2000Sign,fan2018} and principal component estimation in \cite{Marden1999Some,han2018eca}. A critical result is that the spatial Kendall's tau matrix $\Kb$ shares the same ordering of eigenvalues and the same eigenspace as those of the scatter matrix $\bSigma$. We cite this result directly without proof in the following Lemma \ref{lemma:lemma1}.

\begin{lemma}\label{lemma:lemma1}
Let $\bX$  be a continuous elliptically distributed random vector, i.e., $\bX\sim ED(\bmu,\bSigma,\zeta)$ with $\PP(\zeta=0)=0$ and  $\Kb$ be the population multivariate Kendall's tau statistic.  Further assume that $\text{rank}(\bSigma)=q$, we have
$$
\lambda_{j}(\Kb)=\EE\left(\frac{\lambda_{j}(\bSigma)g_j^2}{\lambda_{1}(\bSigma)g_1^2+\cdots+\lambda_{q}(\bSigma)g_q^2}\right),
$$
where $\bg=(g_1,\ldots,g_q)^\top\sim\cN(\zero,\Ib)$, and in addition
$\Kb$ and $\bSigma$ share the same eigenspace with the same descending order of the eigenvalues.
\end{lemma}

The proof of Lemma \ref{lemma:lemma1} can be found in \cite{han2018eca}. By Lemma \ref{lemma:lemma1}, estimating the eigenvectors of $\bSigma$ is equivalent to
estimating those of $\Kb$, and thus $\hat \Kb$ fits the goal of estimating the eigenvectors of $\bSigma$.

\subsection{Robust two-step estimation procedure}\label{sec:2.3}
In this section, we introduce an innovative two-step estimation procedure for large-dimensional elliptical factor model. In the first step, we propose to estimate $\Lb$  by the eigenvectors of the spatial Kendall's tau matrix. First, we estimate the spatial Kendall's tau matrix of $\by_t$ by
\begin{equation}\label{equ:hatKyb}
\hat{\Kb}_{y}=\frac{2}{n(n-1)}\sum_{t<t^\prime}\frac{(\by_t-{\by_{t^\prime}})(\by_t-{\by_{t^\prime}})^\top}{\|\by_t-{\by_{t^\prime}}\|_2^2}.
\end{equation}
As the eigenvectors of the spatial Kendall's tau matrix $\Kb_y$ is identical to the eigenvectors of the scatter matrix $\bSigma_y$, thus we estimate the factor Loading matrix $\Lb$ by $\sqrt p$ times the leading $m$ eigenvectors of $\hat{\Kb}_{y}$. In detail, let $\{\hat{\bxi}_1,\ldots,\hat{\bxi}_m\}$ be the leading $m$ eigenvectors of $\hat \Kb_y$ and let $\hat{\bGamma}=(\hat{\bxi}_1,\ldots,\hat{\bxi}_m)$. We take $\hat\Lb=\sqrt{p}\hat{\bGamma}$ as the estimator of the factor loading matrix $\Lb$. The number of factors $m$ is relatively small compared with $p$ and $n$. We first assume that $m$ is known and fixed. If $m$ is unknown, we can estimate $m$ consistently as in \cite{yu2018robust}.

In a second step, we estimate the factors $\{\bbf_t,t=1,\ldots,n\}$ by regressing $\by_t$ on $\hat\Lb$. $\bbf_t$ is estimated by the following least square optimization,
\begin{equation}\label{equ:fthat}
 \hat{\bbf}_t=\bargmin_{\bbeta_t\in \RR^m}\sum_{i=1}^p\big(y_{it}-\hat{\bl}_i^\top\bbeta_t\big)^2, \ \ t=1,\ldots,n,
\end{equation}
where $\hat \bl_i^\top$ is the $i$-th row of $\hat\Lb$, i.e, $\hat\Lb=(\hat{\bl}_1,\ldots,\hat{\bl}_p)^\top$.
For conventional factor model, when both $n$ and $p$ are large, the factor loadings and the factors can be estimated by PCA, which is equivalent to solving a double least-square regression problem, see \cite{Bai2002Determining} or (2.4) in \cite{fan2018}. The two-step estimation procedure is motivated by the idea of the regression formulation.

\section{Theoretical results}
In this section, we investigate the theoretical properties of the proposed estimators $\hat\Lb$ and $\hat\Fb=(\hat \bbf_1,\ldots,\hat\bbf_n)^\top$. We need the following technical assumptions.

\vspace{1em}
\textbf{Assumption A}  We assume that
	\[
	\left(\begin{aligned}
	&\bbf_t\\
	&\bepsilon_t
	 \end{aligned}\right)=\zeta_t\left(\begin{matrix}
	\Ib_m&{\bf 0}\\
	{\bf 0}&\Ab
	\end{matrix}\right)\frac{\bg_t}{\|\bg_t\|},
	\]
	where   $\zeta_t$'s are independent samples of a scalar random variable $\zeta$, and $\bg_t$'s are independent Gaussian samples from $\bg\sim{\cN}({\bf 0},\Ib_{m+p})$. $m$ is fixed. Further, $\zeta$ and $\bg$ are independent and  $\zeta/\sqrt{p}=O_p(1)$ as $p\rightarrow\infty$. Therefore, $(\bbf_t^\top,\bepsilon_t^\top)^\top$ are independent samples  from $ ED({\bf 0},\bSigma_0,\zeta)$ for $t=1,\ldots,n$ where $\bSigma_0=\left(\begin{matrix}
	\Ib_m&\zero\\
	\zero&\bSigma_\epsilon
	\end{matrix}\right)$, and $\bSigma_{\epsilon}=\Ab\Ab^\top$. To make the model identifiable, we further assume that $\big\|\text{diag}(\bSigma_0)\big\|_{\infty}=1$.
	
	\vspace{1em}
	
	\textbf{Assumption B}  Assume $\Lb^\top\Lb/p\rightarrow\Vb$ as $p\rightarrow\infty$, where $\Vb$ is a positive definite matrix. There exist positive constants $c_1,c_2$ such that $c_2\le \lambda_m(\Vb)<\dots<\lambda_1(\Vb)\le c_1$.
	
	\vspace{1em}
	\textbf{Assumption C}  We assume  $c_2\le \lambda_{\min}(\bSigma_{\epsilon})\le\lambda_{\max}(\bSigma_{\epsilon})\le c_1$.
	\vspace{1em}

\textbf{Assumption A} states that $(\bbf_t^\top,\bepsilon_t^\top)^\top$ are i.i.d and follows the elliptical distribution, which further implies that $\by_t$ are i.i.d. from elliptical distribution with scatter matrix $\bSigma_y=\Lb\Lb^\top+\bSigma_{\epsilon}$.  \textbf{Assumption B} assumes $\Lb^\top \Lb/p$ converges to a positive definite matrix with bounded maximum and minimum eigenvalues.
 We also require that  $\lambda_j({\Vb})$ are distinct  to make corresponding eigenvectors identifiable. \textbf{Assumption C} requires that the eigenvalues of the $\bSigma_{\epsilon}$ are bounded from below and above, which in essence makes the idiosyncratic errors negligible relative to the common component. Compared with the moment assumptions of idiosyncratic errors in \cite{Bai2003Inferential}, \textbf{Assumption C} is another typical way to control the cross-sectional correlations of the idiosyncratic errors. The  assumption is common in related literature, see for example,  \cite{fan2013large} and \cite{fan2018}.
 In fact, under \textbf{Assumption B} and \textbf{Assumption C}, further with the Weyl's theorem, we have that the eigenvalues of $\bSigma_{y}$  show the spiked structure which is a common assumption in the large-dimensional factor model literatures, see for example, \cite{Bai2002Determining,Bai2003Inferential,Ahn2013Eigenvalue,fan2013large,fan2018,Trapani2018A}. In other word, the  eigenvalues $\lambda_1(\bSigma_y),
	\ldots,\lambda_m(\bSigma_y)$ are asymptotically proportional to $p$ while the non-spiked eigenvalues $\lambda_j(\bSigma_y), j>m$ are bounded.

{In the following theorem, we show that the estimated loading matrix converges with the rate $O_p(n^{-1}+p^{-2})$ in terms of the averaged squared error after certain rotation.}

\begin{theorem}\label{theorem:1} Under  \textbf{Assumptions A, B, C},  there exist a series of  matrices $\hat\Hb$ (dependent on $n,p$ and $\hat\Lb$) so that $\hat\Hb^\top\Vb\hat\Hb\overset{p}{\rightarrow}\Ib_m$ and
	$$
	\frac{1}{p} \Big\|\hat\Lb-\Lb\hat\Hb\Big\|_F^2=O_p\Big(\frac{1}{n}+\frac{1}{p^2}\Big).
	$$
\end{theorem}

In Theorem \ref{theorem:1}, we obtain the same convergence rate of the estimated factor loadings as that in \cite{Bai2003Inferential}.
However, we impose no moment constrains on the factors and idiosyncratic errors.
In the following theorem, we establish the convergence rate of the estimated factor scores $\hat \bbf_t$.

\begin{theorem}\label{theorem:2} Assume that  \textbf{Assumptions A, B, C} hold, then for any $t\le n$,
	$$
\|\hat\Hb\hat\bbf_t-\bbf_t\|^2=O_p\Big(\frac{1}{p}+\frac{1}{n^2}\Big).
	$$
\end{theorem}
By the results in Theorem \ref{theorem:1}  and Theorem \ref{theorem:2}, we finally show that the estimated common components are consistent to the true ones.

\begin{theorem}\label{theorem:3}
	Assume that  \textbf{Assumptions A, B, C} hold, we have that for any $t\le n$,
	$$
\frac{1}{p}\Big\|\hat\Lb\hat\bbf_t-\Lb\bbf_t\Big\|^2=O_p\Big(\frac{1}{{n}}+\frac{1}{p}\Big).
	$$
\end{theorem}

As far as we know, this is the first time that consistent estimators for the factor loadings, scores and common components are proposed without any moment constraints. {Under the elliptical assumption containing heavy-tailed cases, our RTS estimators converge at the same rates as those of the PCA estimators with finite fourth moment constrains on the factors and errors, see \cite{Bai2003Inferential}.}

\section{Simulation Study}\label{ss}

	In this section, we conduct thorough  simulation studies to compare the Robust Two-Step (RTS) estimator with the conventional PCA method.
	We use similar data-generating models as in \cite{Ahn2013Eigenvalue}, \cite{Xia2017Transformed} and \cite{yu2018robust}. We generate the data from the following model,
	\begin{align}
	 &y_{it}=\sum\limits_{j=1}^{m}L_{ij}f_{jt}+\sqrt{\theta}u_{it},\quad u_{it}=\sqrt{\frac{1-\rho^2}{1+2J\beta^2}}e_{it}, \nonumber \\
&e_{it}=\rho e_{i,t-1}+(1-\beta)v_{it}+\sum_{l={\rm max}\{i-J,1\}}^{{\rm min}\{i+J,p\}}\beta v_{lt}, \ \ t=1,\ldots,n, \ \ i=1,\ldots,p,\nonumber
	\end{align}
	where  $\bbf_t=(f_{1t},\ldots,f_{mt})^\top$ and $\bv_t=(v_{1t},\ldots,v_{pt})^\top$ are jointly generated from  elliptical distributions. We let $L_{ij}$ be independently drawn from the standard normal distribution. The parameter $\theta$ controls the SNR (signal to noise ratio), $\rho$ controls the serial correlations of idiosyncratic errors, and $\beta$ and $J$ control the cross-sectional correlations. We point out that although we  assume $\by_t$'s are temporally independent theoretically in \textbf{Assumption A}, we allow $\bu_t$ to be serially correlated in the simulation studies.

\begin{table}[!h]
  \caption{Simulation results for Scenario A, the values in the parentheses are the interquartile ranges for MEE-CC and standard deviations for AVE-FL and AVE-FS.}
  \label{tab:1}
  \renewcommand{\arraystretch}{1.3}
  \centering
  \selectfont
  \begin{threeparttable}
   \scalebox{0.9}{ \begin{tabular*}{16cm}{ccccccccccccccccccccccccccccc}
\toprule[2pt]
&\multirow{2}{*}{Type}&\multirow{2}{*}{Method}  &\multicolumn{3}{c}{$(p,n)=(150,100)$}&\multicolumn{3}{c}{$(p,n)=(250,100)$} \cr
\cmidrule(lr){4-6} \cmidrule(lr){7-9}
&&                 &$\text{MEE\_CC}$     &$\text{AVE\_FL}$      &$\text{AVE\_FS}$      &$\text{MEE\_CC}$      &$\text{AVE\_FL}$      &$\text{AVE\_FS}$  \\
\midrule[1pt]
&$\mathcal{N}(\zero,\Ib_{p+m})$   &RTS     &0.02(0.00)	&0.11(0.01)	&0.08(0.01)	&0.01(0.00)	 &0.11(0.01)	&0.06(0.00) \\
&                                 &PCA     &0.02(0.00)	&0.10(0.01)	&0.08(0.01)	&0.01(0.00)	 &0.10(0.01)	&0.06(0.00) \\
\cmidrule(lr){4-9}
&$t_{3}(\zero,\Ib_{p+m})$         &RTS     &0.02(0.00)	&0.12(0.01)	&0.08(0.01)	&0.02(0.00)	 &0.11(0.01)	&0.07(0.01) \\
&                                 &PCA     &0.04(0.03)	&0.20(0.06)	&0.10(0.04)	&0.04(0.03)	 &0.20(0.06)	&0.08(0.04) \\
\cmidrule(lr){4-9}
&$t_{2}(\zero,\Ib_{p+m})$         &RTS     &0.02(0.01)  &0.12(0.01)	&0.09(0.02)	 &0.02(0.00)    &0.12(0.01)	&0.07(0.01) \\
&                                 &PCA     &0.09(0.11)	&0.30(0.10)	&0.14(0.09)	 &0.09(0.10) 	 &0.29(0.10)	&0.12(0.08) \\
\cmidrule(lr){4-9}
&$t_{1}(\zero,\Ib_{p+m})$         &RTS     &0.02(0.01)    &0.12(0.01)	&0.09(0.04)	 &0.02(0.01)          &0.12(0.01)	&0.07(0.03) \\
&                                 &PCA     &0.29(0.29)	&0.52(0.12)	&0.27(0.16)	 &0.29(0.29)  	 &0.52(0.12)	&0.26(0.16) \\
\cmidrule(lr){4-9}
&Skewed $t_{3}$                     &RTS     &0.02(0.00)  &0.12(0.01)	&0.08(0.01)	 &0.02(0.00)    &0.11(0.01)	&0.07(0.01) \\
&                                 &PCA     &0.04(0.03)	&0.20(0.06)	&0.10(0.03)	 &0.04(0.03) 	 &0.20(0.06)	&0.09(0.04) \\
\cmidrule(lr){4-9}
&$\alpha$-stable                  &RTS     &0.06(0.02)  &0.19(0.01)	&0.19(0.06)	 &0.06(0.02)    &0.19(0.01)	&0.16(0.06) \\
&                                 &PCA     &0.14(0.70)	&0.39(0.21)	&0.37(0.22)	 &0.16(0.80)  	 &0.41(0.21)	&0.37(0.23) \\
\midrule[1pt]
&\multirow{2}{*}{Type}&\multirow{2}{*}{Method}  &\multicolumn{3}{c}{$(p,n)=(250,150)$}&\multicolumn{3}{c}{$(p,n)=(250,200)$} \cr
\cmidrule(lr){4-6} \cmidrule(lr){7-9}
&&                   &$\text{MEE\_CC}$     &$\text{AVE\_FL}$      &$\text{AVE\_FS}$      &$\text{MEE\_CC}$      &$\text{AVE\_FL}$      &$\text{AVE\_FS}$  \\
\midrule[1pt]
&$\mathcal{N}(\zero,\Ib_{p+m})$  &RTS    &0.01(0.00)	&0.09(0.00)	&0.06(0.00)	&0.01(0.00)	 &0.07(0.00)	&0.06(0.00)	\\
&                                &PCA    &0.01(0.00)	&0.08(0.00)	&0.06(0.00)	&0.01(0.00)	 &0.07(0.00)	&0.06(0.00)  \\
\cmidrule(lr){4-9}
&$t_{3}(\zero,\Ib_{p+m})$        &RTS    &0.01(0.00)	&0.09(0.00)	&0.06(0.01)	&0.01(0.00)	 &0.08(0.00)	&0.06(0.01) 	\\
&                                &PCA    &0.03(0.02)	&0.18(0.06)	&0.08(0.04)	&0.03(0.02)	 &0.16(0.05)	&0.08(0.02)\\
\cmidrule(lr){4-9}
&$t_{2}(\zero,\Ib_{p+m})$        &RTS    &0.01(0.00)    &0.10(0.01)	&0.07(0.01)	 &0.01(0.00) &0.08(0.00)	&0.06(0.01) \\
&                                &PCA    &0.08(0.10)	&0.27(0.10)	&0.12(0.08)	 &0.08(0.09) &0.27(0.10)	&0.11(0.08)     \\
\cmidrule(lr){4-9}
&$t_{1}(\zero,\Ib_{p+m})$        &RTS    &0.01(0.01)    &0.10(0.01)	&0.07(0.03)	 &0.01(0.00) &0.09(0.00)	&0.07(0.04) \\
&                                &PCA    &0.28(0.31)    &0.52(0.12)	&0.27(0.16)	 &0.28(0.27) &0.51(0.12)	&0.26(0.16)  \\
\cmidrule(lr){4-9}
&Skewed $t_{3}$                    &RTS     &0.01(0.00)  &0.09(0.00)	&0.06(0.01)	 &0.01(0.00)    &0.08(0.00)	&0.06(0.01) \\
&                                &PCA     &0.03(0.03)	&0.18(0.05)	&0.08(0.03)	 &0.03(0.02) 	 &0.16(0.06)	&0.08(0.04) \\
\cmidrule(lr){4-9}
&$\alpha$-stable                 &RTS     &0.04(0.01)    &0.16(0.01)	&0.15(0.05)	 &0.04(0.01)&0.14(0.01)	&0.15(0.06) \\
&                                &PCA     &0.15(0.82)	&0.39(0.23)	&0.37(0.24)	 &0.13(0.83)  	 &0.39(0.24)	&0.38(0.24) \\
\bottomrule[2pt]
  \end{tabular*}}
  \end{threeparttable}
\end{table}

Before we give the data generating scenarios, we first review the multivariate $t$ distribution.  The Probability Distribution Function (PDF) of a $d$-dimensional multivariate $t$ distribution $t_{\nu}(\bmu,\bSigma_{d\times d})$ is
	\begin{displaymath}
	 \frac{{\Gamma\big((\nu+d)/2\big)}}{\Gamma(\nu/2)\nu^{d/2}\pi^{d/2}|\bSigma|^{1/2}}\bigg\{1+\frac{1}{\nu}(\bx-\bmu)^\top\bSigma^{-1}(\bx-\bmu)\bigg\}^{-(\nu+d)/2},
	\end{displaymath}
	 {where $\Gamma(\cdot)$ is the gamma function}. In fact, multivariate $t$ distribution with  $\nu=1$ is the multivariate Cauchy distribution that has no finite mean. We also consider the following data generating scenarios in the simulation studies.

\vspace{0.5em}

	\textbf {Scenario A} Set $m=3,\theta=1,\rho=\beta=J=0$, $(p,n)=\big\{(150,100),(250,100),(250,150),(250,200)\big\}$, $(\bbf_t^\top,\bv_t^\top)^\top$ are generated in the following ways: (i) $(\bbf_t^\top,\bv_t^\top)^\top$ are \emph{i.i.d.} jointly elliptical random samples from multivariate Gaussian distributions $\mathcal{N}(\zero,\Ib_{p+m})$; (ii) $(\bbf_t^\top,\bv_t^\top)^\top$ are \emph{i.i.d.} jointly elliptical random samples from  multivariate centralized $t$ distributions $t_{\nu}(\zero,\Ib_{p+m})$ with $\nu=3, 2, 1$; (iii) $(\bbf_t^\top,\bv_t^\top)^\top$ are \emph{i.i.d.} samples from  multivariate skewed-$t_3$ distribution; (iv) $\bbf_t^\top$ are \emph{i.i.d.} random samples from multivariate Gaussian distributions $\mathcal{N}(\zero,\Ib_{m})$ while the elements of  $\bv_t$ are \emph{i.i.d.} random samples from symmetric $\alpha$-Stable distribution $S_{\alpha}(\beta, \gamma, \gamma)$ with skewness parameter
$\beta=0$, scale parameter $\gamma=1$ and location parameter $\delta=0$, $\alpha=1.8$.

\vspace{0.5em}

	\textbf {Scenario B} Set $m=3,\theta=1,\rho=0.5,\beta=0.2,J={\rm max}\{10, p/20\}$, $(\bbf_t^\top,\bv_t^\top)$ are generated in the same ways as in Scenario A. $(p,n)=\{(150,100),(250,100),(250,150),(250,200)\}$.

\vspace{0.5em}

	\textbf {Scenario C} Set $m=3,\theta=1,\rho=0.5,\beta=0.2,J={\rm max}\{10, p/20\}, (p,n)=(250,100)$, $(\bbf_t^\top,\bv_t^\top)$ are \emph{i.i.d.} jointly elliptical random vectors from multivariate Gaussian $\mathcal{N}(\zero,\Db)$ and multivariate centralized $t$ distribution $t_{\nu}(\zero,\Db)$ with $\nu=3$, where {$\Db$} is a diagonal matrix of dimension $(p+m)\times(p+m)$  with $\mathrm{D}_{ii}=1,i\neq 3$ and $\mathrm{D}_{33}=\text{SNR}$ with $\text{SNR}$ from $\{0.7,0.6,0.5,0.4\}$.

\vspace{0.5em}

\begin{table}[!h]
  \caption{Simulation results for Scenario B, the values in the parentheses are the interquartile ranges for MEE-CC and standard deviations for AVE-FL, AVE-FS.}
  \label{tab:2}
  \renewcommand{\arraystretch}{1.4}
  \centering
  \selectfont
  \begin{threeparttable}
   \scalebox{0.9}{ \begin{tabular*}{16cm}{ccccccccccccccccccccccccccccc}
\toprule[2pt]
&\multirow{2}{*}{Type}&\multirow{2}{*}{Method}  &\multicolumn{3}{c}{$(p,n)=(150,100)$}&\multicolumn{3}{c}{$(p,n)=(250,100)$} \cr
\cmidrule(lr){4-6} \cmidrule(lr){7-9}
&&                 &$\text{MEE\_CC}$     &$\text{AVE\_FL}$      &$\text{AVE\_FS}$      &$\text{MEE\_CC}$      &$\text{AVE\_FL}$      &$\text{AVE\_FS}$  \\
\midrule[1pt]
&$\mathcal{N}(\zero,\Ib_{p+m})$  &RTS    &0.02(0.01)	&0.11(0.02)	&0.09(0.01)	&0.02(0.00)	 &0.11(0.01)	&0.07(0.01)   \\
&                                &PCA   &0.02(0.01)	 &0.11(0.02)	&0.09(0.01)	&0.02(0.00)	 &0.11(0.01)	 &0.07(0.01)   \\
\cmidrule(lr){4-9}
&$t_{3}(\zero,\Ib_{p+m})$        &RTS    &0.02(0.01)	&0.13(0.02)	&0.10(0.02)	&0.02(0.01)	 &0.13(0.02)	&0.08(0.02)   \\
&                                &PCA   &0.04(0.03)	 &0.20(0.07)	&0.12(0.07)	&0.04(0.03)	 &0.20(0.07)	 &0.10(0.06)   \\
\cmidrule(lr){4-9}
&$t_{2}(\zero,\Ib_{p+m})$        &RTS    &0.03(0.02)    &0.15(0.03)	&0.12(0.05)	&0.02(0.01)	 &0.14(0.03)	&0.09(0.05)  \\
&                                &PCA   &0.09(0.11)	 &0.31(0.13)	&0.21(0.14)	&0.08(0.10)  &0.29(0.13)	 &0.18(0.14)  \\
\cmidrule(lr){4-9}
&$t_{1}(\zero,\Ib_{p+m})$        &RTS    &0.08(0.16)	&0.29(0.14)	&0.28(0.17)	 &0.06(0.11)	    &0.27(0.14)	&0.25(0.18)  \\
&                                &PCA   &0.28(0.30)     &0.55(0.13)	&0.44(0.17)	 &0.28(0.32)	&0.55(0.13)	 &0.43(0.18)  \\
\cmidrule(lr){4-9}
&Skewed $t_{3}$                     &RTS     &0.02(0.01)  &0.12(0.02)	&0.10(0.02)	 &0.02(0.01)    &0.12(0.01)	&0.08(0.02) \\
&                                 &PCA     &0.05(0.04)	&0.21(0.07)	&0.13(0.07)	 &0.04(0.03) 	 &0.20(0.07)	&0.11(0.07) \\
\cmidrule(lr){4-9}
&$\alpha$-stable                  &RTS     &0.14(0.14)    &0.30(0.10)	&0.30(0.13)	 &0.12(0.11)&0.29(0.09)	&0.27(0.13) \\
&                                 &PCA     &0.35(0.78)	&0.46(0.18)	&0.43(0.20)	 &0.43(0.86)	 &0.48(0.19)	&0.43(0.21) \\
\midrule[1pt]
&\multirow{2}{*}{Type}&\multirow{2}{*}{Method}  &\multicolumn{3}{c}{$(p,n)=(250,150)$}&\multicolumn{3}{c}{$(p,n)=(250,200)$} \cr
\cmidrule(lr){4-6} \cmidrule(lr){7-9}
&&                   &$\text{MEE\_CC}$     &$\text{AVE\_FL}$      &$\text{AVE\_FS}$      &$\text{MEE\_CC}$      &$\text{AVE\_FL}$      &$\text{AVE\_FS}$  \\
\midrule[1pt]
&$\mathcal{N}(\zero,\Ib_{p+m})$  &RTS    &0.01(0.00)	&0.09(0.01)	&0.07(0.01)	&0.01(0.00)	 &0.08(0.01)	&0.07(0.01) 	\\
&                                 &PCA   &0.01(0.00)	&0.09(0.01)	&0.07(0.01)	&0.01(0.00)	 &0.07(0.01)	&0.07(0.01)	  \\
\cmidrule(lr){4-9}
&$t_{3}(\zero,\Ib_{p+m})$        &RTS    &0.01(0.00)	&0.10(0.01)	&0.07(0.01)	&0.01(0.00)	 &0.09(0.01)	&0.07(0.01)  	\\
&                                 &PCA   &0.03(0.02)	&0.17(0.06)	&0.09(0.05)	&0.03(0.02)	 &0.15(0.06)	&0.09(0.04)	\\
\cmidrule(lr){4-9}
&$t_{2}(\zero,\Ib_{p+m})$        &RTS    &0.02(0.01)    &0.11(0.01)	&0.08(0.03)	 &0.01(0.00) &0.10(0.01)	&0.08(0.03)   \\
&                                 &PCA   &0.07(0.10)	&0.28(0.13)	&0.17(0.14)	 &0.07(0.09) &0.27(0.13)	&0.16(0.14)     \\
\cmidrule(lr){4-9}
&$t_{1}(\zero,\Ib_{p+m})$         &RTS   &0.04(0.04)	    &0.19(0.09)	&0.19(0.15)	 &0.02(0.02)    &0.14(0.04)	&0.17(0.14)  \\
&                                 &PCA   &0.27(0.31) &0.55(0.13)	&0.43(0.18)	 &0.27(0.28) &0.54(0.13)	&0.42(0.18)  \\
\cmidrule(lr){4-9}
&Skewed $t_{3}$                     &RTS     &0.01(0.00)  &0.10(0.01)	&0.07(0.02)	 &0.01(0.00)    &0.09(0.01)	&0.07(0.01) \\
&                                 &PCA     &0.03(0.02)	&0.17(0.06)	&0.10(0.04)	 &0.03(0.02) 	 &0.16(0.06)	&0.09(0.05) \\
\cmidrule(lr){4-9}
&$\alpha$-stable                  &RTS     &0.08(0.06)    &0.22(0.04)	&0.23(0.10)	 &0.06(0.03) &0.17(0.02) &0.20(0.08) \\
&                                 &PCA     &0.41(0.88)	&0.46(0.21)	&0.43(0.22)	 & 0.33(0.86) 	 &0.43(0.22)	&0.42(0.23) \\
\bottomrule[2pt]
  \end{tabular*}}
  \end{threeparttable}
\end{table}

In \textbf {Scenario A}, the setting perfectly fits to our assumption with no serial correlations of idiosyncratic errors and $(\bbf_t^\top,\bv_t^\top)^\top$ are  from light-tailed  Gaussian $\mathcal{N}(\zero,\Ib_{p+m})$ or heavy-tailed $t_{\nu}(\zero,\Ib_{p+m})$ with $\nu=3, 2, 1$.  Note that when $\nu=1$, it's indeed the Cauchy distribution which does not have finite moments of order greater than or equal to one. We also consider the skewed $t_3$ and $\alpha$-stable distributions to  gauge how sensitive the method is to the elliptical distribution assumption. We generate the multivariate skewed $t_3$ random samples from $\mathcal{ST}_{N+r}(\bxi={\bf 0},\bOmega=\Ib,\balpha={\bf 20},\nu=3)$ by function \texttt{rmvst} in \textsf{R} package \texttt{fMultivar}.
\textbf{Scenario B} is a simple case containing both serially and cross-sectionally correlated errors from Gaussian distribution, $t$ distribution with degree 1,2,3, skewed $t_3$ distribution or $\alpha$-stable distribution.  \textbf{Scenario C} corresponds to a  case where both serially and cross-sectionally correlated errors, and strong and weak factors exist.
To evaluate the empirical performance of different methods, we consider the following indices: the \textbf{ME}dian of the normalized estimation \textbf{E}rrors for   \textbf{C}ommon \textbf{C}omponents in terms of the matrix Frobenius norm, denoted as MEE-CC; the \textbf{AV}erage estimation \textbf{E}rror for the \textbf{F}actor \textbf{L}oading matrices, denoted as AVE-FL; and the \textbf{AV}erage estimation \textbf{E}rror for the \textbf{F}actor \textbf{S}croe matrices, denoted as AVE-FS. Specifically,
\[
\text{MEE-CC}=\text{median}\left\{\|\hat\Lb_r\hat{\Fb}_r^\top-\Lb\Fb^\top\|_F^2/\|\Lb\Fb^\top\|_F^2,r=1,\ldots,R\right\},
\]
\[
\text{AVE-FL}=\frac{1}{R}\sum_{r=1}^R \cD(\hat \Lb_r,\Lb), \ \ \text{AVE-FS}=\frac{1}{R}\sum_{r=1}^R \cD(\hat \Fb_r,\Fb),
\]
where $R$ is the replication times, $\hat \Lb_r$ and $\hat \Fb_r$ are the estimators for the $r$th replication, and for {two orthogonal matrices} $\Ob_1$ and $\Ob_2$ of sizes $p\times q_1$ and $p\times q_2$,
\[
\cD(\Ob_1,\Ob_2)=\bigg(1-\frac{1}{\max{(q_1,q_2)}}\text{Tr}\Big(\Ob_1\Ob_1^\top\Ob_2\Ob_2^\top\Big)\bigg)^{1/2}.
\]

\begin{table}[h]
  \caption{Simulation results for Scenario C, the values in the parentheses are the interquartile ranges for MEE-CC and standard deviations for AVE-FL and AVE-FS.}
  \label{tab:3}
  \renewcommand{\arraystretch}{1.4}
  \centering
  \selectfont
  \begin{threeparttable}
    \scalebox{0.9}{
    	 \begin{tabular*}{17cm}{ccccccccccccccccccccccccccccc}
\toprule[2pt]
&\multirow{2}{*}{SNR}  &\multirow{2}{*}{Type}&\multirow{2}{*}{Method}  &\multicolumn{3}{c}{$(p,n)=(250,100)$}&\multicolumn{3}{c}{$(p,n)=(250,150)$} \cr
\cmidrule(lr){5-7} \cmidrule(lr){8-10}
&&&                 &$\text{MEE\_CC}$     &$\text{AVE\_FL}$      &$\text{AVE\_FS}$      &$\text{MEE\_CC}$      &$\text{AVE\_FL}$      &$\text{AVE\_FS}$  \\
\midrule[1pt]
&\multirow{4}{*}{0.4}
&\multirow{2}{*}{$\mathcal{N}(\zero,\Ib_{p+m})$} &RTS   &0.02(0.01)	&0.14(0.02)	&0.09(0.02)	 &0.02(0.01)	 &0.11(0.01)	&0.09(0.01)  \\
&&                                               &PCA   &0.02(0.01)	&0.13(0.02)	&0.09(0.01)	 &0.02(0.01)	 &0.11(0.01)	&0.09(0..01)  \\
\cmidrule(lr){5-10}
&&\multirow{2}{*}{$t_{3}(\zero,\Ib_{p+m})$}      &RTS   &0.03(0.01)	&0.16(0.03)	&0.11(0.03)	 &0.02(0.01)	 &0.13(0.02)	&0.10(0.02)   \\
&&                                               &PCA   &0.05(0.06)	&0.26(0.10)	&0.16(0.10)	 &0.04(0.04)	 &0.23(0.10)	&0.15(0.10)   \\
\midrule[1pt]
&\multirow{4}{*}{0.5}
&\multirow{2}{*}{$\mathcal{N}(\zero,\Ib_{p+m})$} &RTS   &0.02(0.01)	&0.13(0.02)	&0.08(0.01)	 &0.02(0.01)	 &0.10(0.01)	&0.08(0.01)  \\
&&                                               &PCA   &0.02(0.01)	&0.13(0.02)	&0.08(0.01)	 &0.01(0.01)	 &0.10(0.01)	&0.08(0.01)  \\
\cmidrule(lr){5-10}
&&\multirow{2}{*}{$t_{3}(\zero,\Ib_{p+m})$}      &RTS   &0.03(0.01)	&0.15(0.02)	&0.10(0.03)	 &0.02(0.01)	 &0.12(0.01)	&0.09(0.02) \\
&&                                               &PCA   &0.05(0.04)	&0.23(0.09)	&0.14(0.09)	 &0.04(0.03)	 &0.21(0.09)	&0.13(0.08)  \\
\midrule[1pt]
&\multirow{4}{*}{0.6}
&\multirow{2}{*}{$\mathcal{N}(\zero,\Ib_{p+m})$} &RTS   &0.02(0.00)	&0.12(0.01)	&0.08(0.01)	 &0.01(0.01)	 &0.10(0.01)	&0.08(0.01)  \\
&&                                               &PCA   &0.02(0.00)	&0.12(0.01)	&0.08(0.01)	 &0.01(0.01)	 &0.10(0.01)	&0.08(0.01)  \\
\cmidrule(lr){5-10}
&&\multirow{2}{*}{$t_{3}(\zero,\Ib_{p+m})$}      &RTS   &0.02(0.01)	&0.14(0.02)	&0.09(0.03)	 &0.02(0.01)	 &0.11(0.01)	&0.08(0.02)  \\
&&                                               &PCA   &0.05(0.04)	&0.22(0.08)	&0.12(0.08)	 &0.03(0.03)	 &0.19(0.08)	&0.11(0.07)  \\
\midrule[1pt]
&\multirow{4}{*}{0.7}
&\multirow{2}{*}{$\mathcal{N}(\zero,\Ib_{p+m})$} &RTS   &0.02(0.00)	&0.12(0.01)	&0.08(0.01)	 &0.01(0.01)	 &0.09(0.01)	&0.07(0.01)  \\
&&                                               &PCA   &0.02(0.00)	&0.11(0.01)	&0.08(0.01)	 &0.01(0.01)	 &0.09(0.01)	&0.07(0.01)  \\
\cmidrule(lr){5-10}
&&\multirow{2}{*}{$t_{3}(\zero,\Ib_{p+m})$}      &RTS   &0.02(0.01)	&0.14(0.02)	&0.08(0.02)	 &0.02(0.00)	 &0.11(0.01)	&0.08(0.01)  \\
&&                                               &PCA   &0.05(0.04)	&0.21(0.08)	&0.12(0.07)	 &0.03(0.02)	 &0.18(0.07)	&0.10(0.06)  \\
\bottomrule[2pt]
  \end{tabular*}
}
  \end{threeparttable}
\end{table}

The Gram-Schmidt orthonormal transformation  can be used when $\Ob_1$ and $\Ob_2$ are not column-orthogonal matrices. In fact, $\cD(\Ob_1,\Ob_2)$ measures the distance between the column spaces of  $\Ob_1$ and $\Ob_2$, and it is a quantity between 0 and 1. It is equal to 0 if the column spaces of $\Ob_1$ and $\Ob_2$  are the same and 1 if they are orthogonal. As the factor loading matrix and factor score matrix are not separately identifiable, $\cD(\cdot,\cdot)$ particularly suits to quantify   the accuracy of factor loading/score matrices estimation.

The simulation results for Scenario A,  Scenario B and  Scenario C are reported in Table \ref{tab:1}, Table \ref{tab:2} and Table \ref{tab:3}, respectively.
For Scenario A, from Table \ref{tab:1}, we can see that in Gaussian setting, PCA performs slightly better than the RTS in terms of $\text{MEE\_CC}$, $\text{AVE\_FL}$ and $\text{AVE\_FS}$ while the performances of the two methods are still comparable. In the heavy-tailed $t_{\nu}$ distribution settings with $\nu=1,2,3$, the RTS outperforms the PCA by a large margin,  in terms of $\text{MEE\_CC}$, $\text{AVE\_FL}$ and $\text{AVE\_FS}$, which indicates the robustness of the RTS procedure. The smaller the $\nu$ is, the more obvious the superiority of RTS over PCA is.
Besides, as the time dimension $n$ or the cross-section $p$ get larger, the RTS performs better. In the skewed-$t_3$ and $\alpha$-stable distribution settings, the RTS still outperforms the PCA and performs satisfactorily, which indicates the RTS is not sensitive to the symmetric elliptical distribution assumption.
For Scenario B, from Table \ref{tab:2}, we can draw similar conclusions as for Scenario A. The results show that the proposed RTS procedure is also robust to the heavy tails in the case where both serial and cross-sectional correlations exist. For Scenario C,  from Table \ref{tab:3}, we see that RTS still performs well when both strong and weak factors exist. It performs comparably with PCA for data from $\cN(\zero, \Db)$ while performs much better than PCA for heavy-tailed data from $t_3(\zero,\Db)$ distribution.
In addition, the performances of both the RTS and PCA methods become better as the time dimension $n$ gets larger (from 100 to 150). In a word, from the simulation study,  the proposed RTS procedure can be used as a safe replacement of the conventional PCA method in practice.

\section{Real Example: S$\&$P 100 Weekly Returns Panel}
	In this section, we apply the proposed method to study the weekly share returns of Standard \& Poor 100 component companies during the period from January 1st, 2018 to December 31st, 2019. Details of the data set are available upon request, including the symbol list and names of the corresponding companies. The raw data set is a ``105 weeks''$\times$``100 shares'' panel without missing values. We firstly calculate the sample auto-correlation functions, which indicates that no significant serial correlations exist for most of the weekly return series. We also performed the Augmented Dickey-Fuller tests and found that  all the series are stationary.
	
	We use the centralized log returns to do factor analysis. We leave the centralized log returns unscaled since volatilities of all assets are themselves very informative in portfolio allocation, risk management and derivatives pricing. As for the factor number, the ``eigenvalue-ratio'' method proposed by \cite{Ahn2013Eigenvalue} and its robust version proposed by \cite{yu2018robust} both lead to an estimate of just 1 common factor. Inspired by the Fama-French 3 factor model, we also consider $m=3$ in this example for comparison.

\begin{figure}[h]
	\centering
	\scalebox{0.8}{
		\begin{minipage}[!t]{0.48\linewidth}
			 \includegraphics[width=1\textwidth]{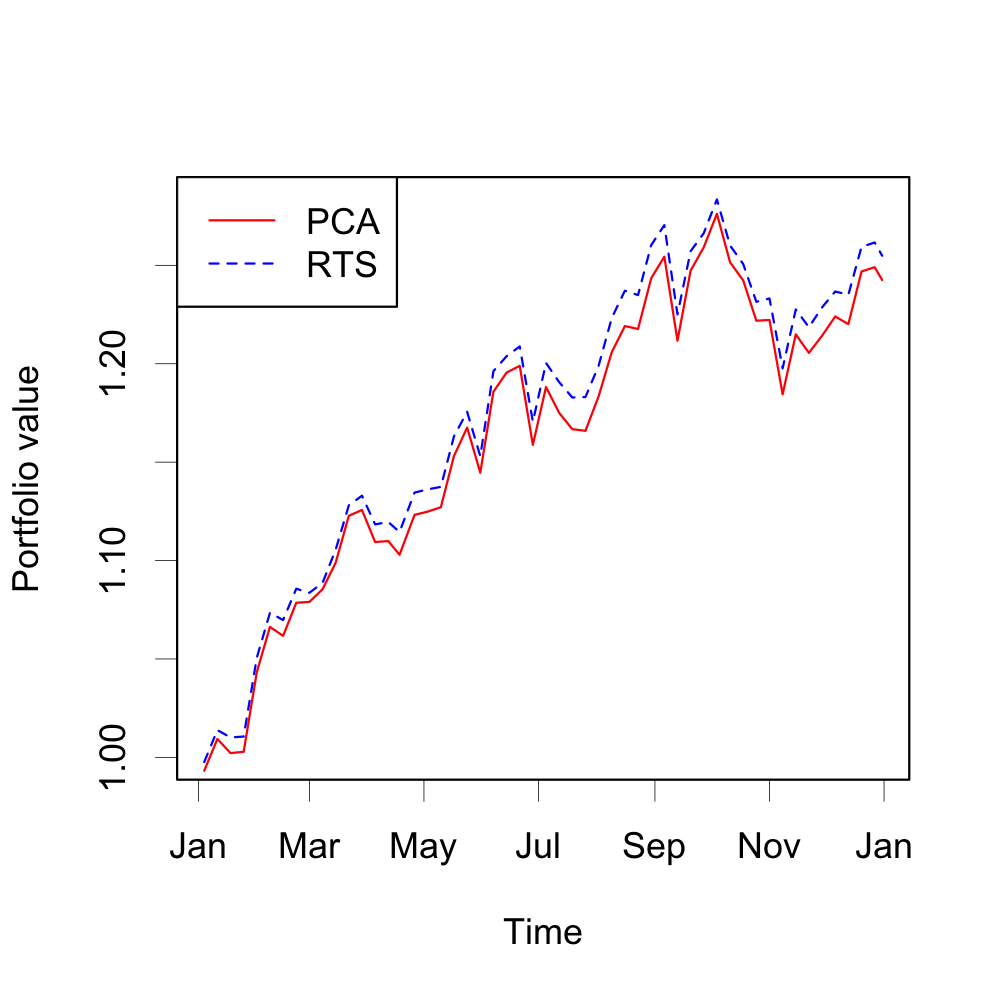}\\
		\end{minipage}
		\begin{minipage}[!t]{0.48\linewidth}
			 \includegraphics[width=1\textwidth]{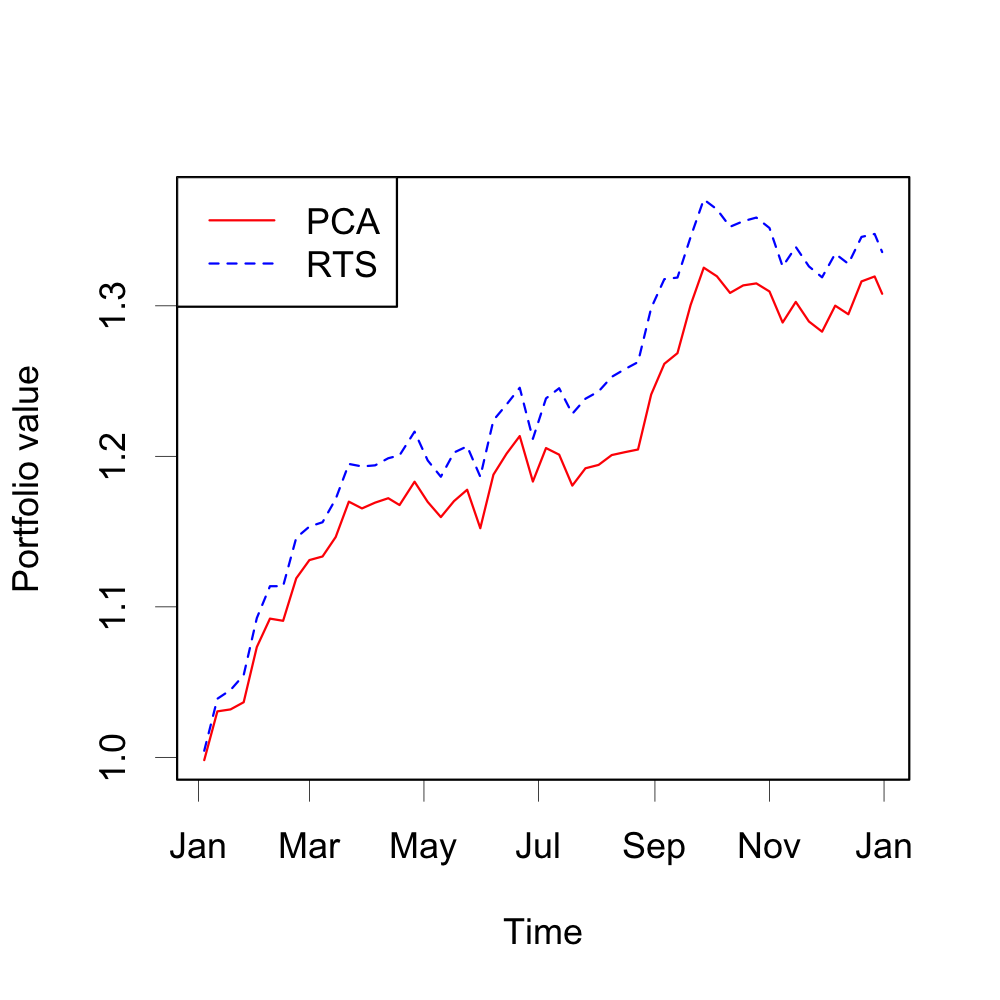}\\
	\end{minipage}}
	\caption{Net value curves of the portfolio when using PCA and RTS to estimate the factor model with number of factors $m=1$ (left) and $m=3$ (right).}\label{fig:2}
\end{figure}

	Firstly, we design a rolling scheme to evaluate the PCA and RTS methods based on their performances in the annual return by constructing risk-minimization portfolios. Given the scatter matrix $\bSigma$ of the share returns, a risk-minimization portfolio is given by
	\[
 \bw_{\text{opt}}=\bargmin_{{\bf 1}^\top\bw=1}\bw^\top\bSigma\bw=\frac{\bSigma^{-1}{\bf 1}}{{\bf 1}^\top\bSigma^{-1}{\bf 1}},
\]
where $\bw_{\text{opt}}$ determines the optimal weights on the shares.  Since the scatter matrix $\bSigma$ is unknown in practice, we use the factor model with PCA or RTS method to estimate it. In detail, at the beginning of each week $t$ in the year 2019, we recursively use the returns during the past 52 weeks ($52\times 100$ panel) to train  factor models either by PCA or RTS method. The estimated common components and idiosyncratic errors are recorded as $\hat\bchi_t$ and $\hat\Eb_t$, both of dimension $52\times 100$. Then, we ignore the cross-sectional correlations of the idiosyncratic errors and estimate the scatter matrix of the 100 variables at week $t$ by
	\[
	 \hat{\bSigma}_t=\frac{1}{52}\hat\bchi_t^\top\hat\bchi_t+\text{diag}\bigg(\frac{1}{52}\hat\Eb_t^\top\hat\Eb_t\bigg).
	\]
The portfolio weights $\hat\bw_t$ are calculated using $\hat\bSigma_t$ and the return of the portfolio at week $t$ is $\hat\bw_t^\top\bx_t$, where $\bx_t$ is the raw return vector at week $t$. Figure \ref{fig:2} shows the net value curves of this strategy during the year 2019 by ignoring transaction cost and liquidity risk. It can be seen that when the factor models are trained by RTS method, the annual return of this portfolio is higher than that by PCA, regardless of taking $m=1$ or $m=3$.

\begin{figure}[h]
		\centering
		\scalebox{0.8}{
		\begin{minipage}[!t]{0.48\linewidth}
			 \includegraphics[width=1\textwidth]{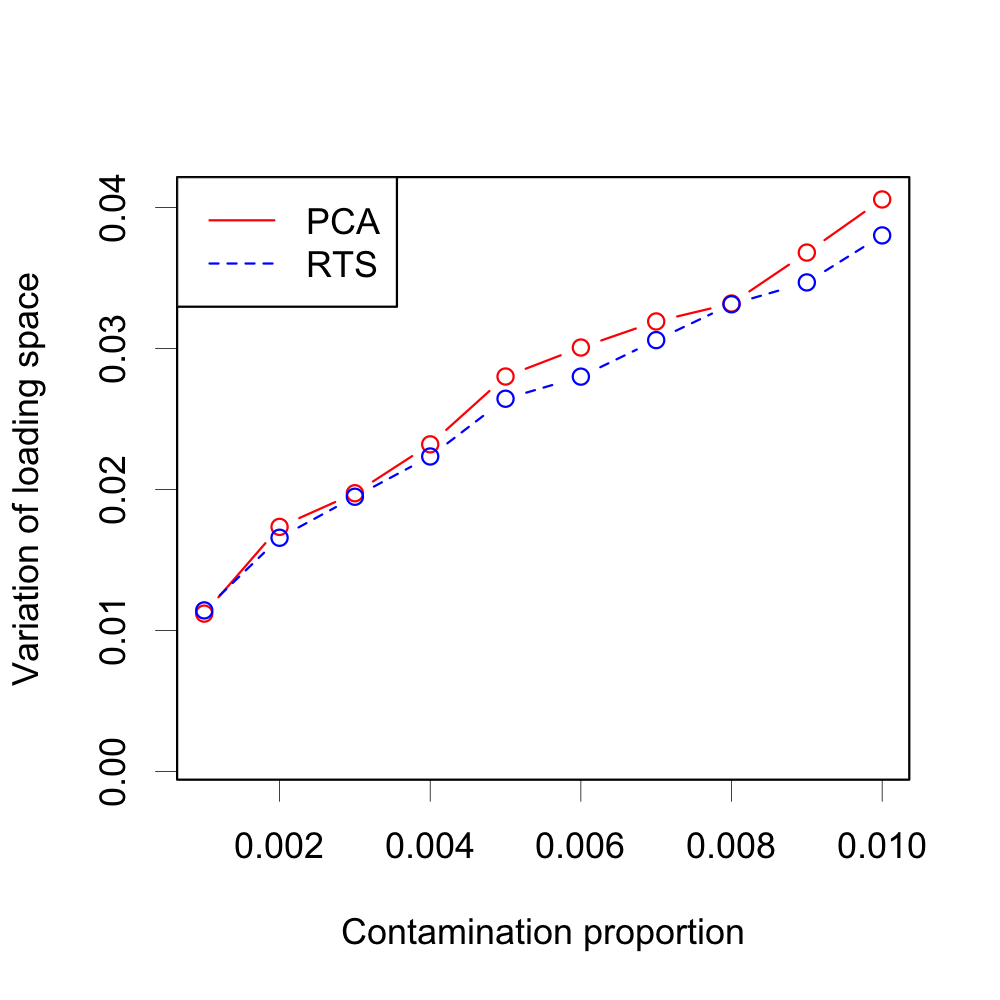}\\
		\end{minipage}
		\begin{minipage}[!t]{0.48\linewidth}
			 \includegraphics[width=1\textwidth]{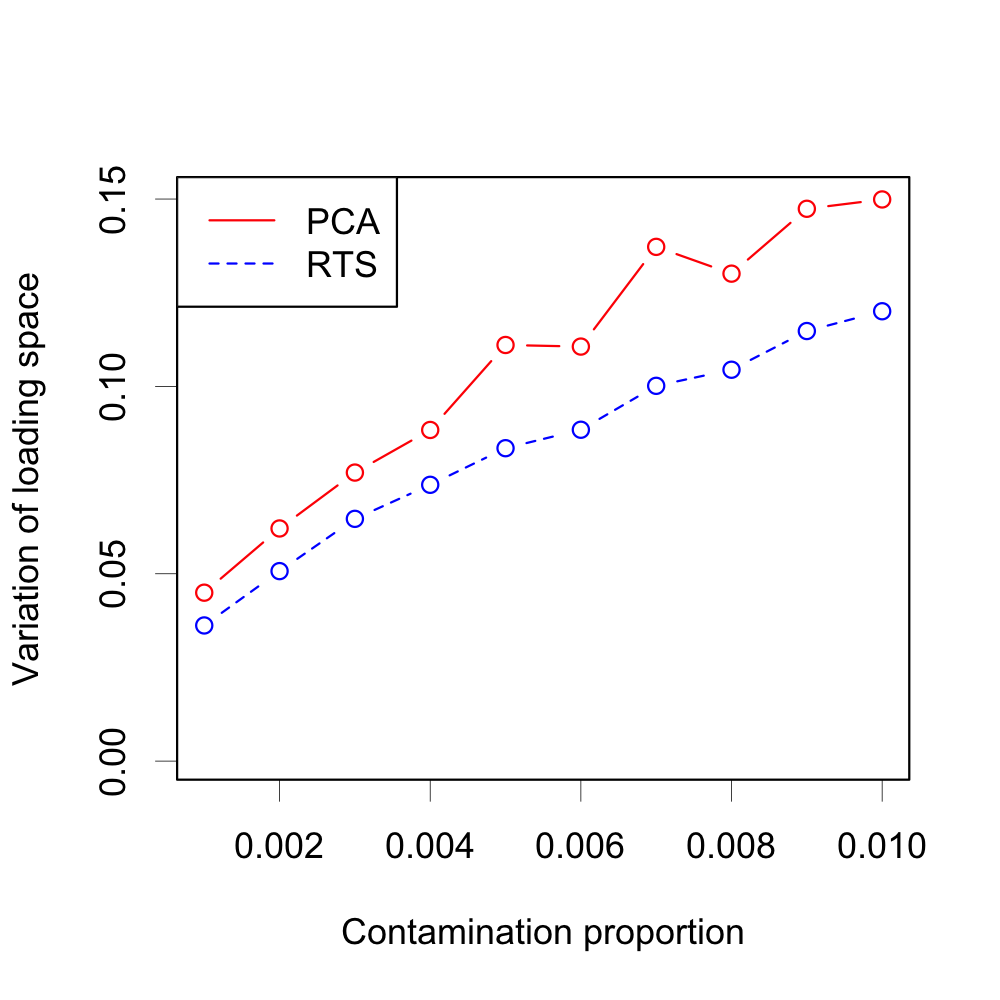}\\
		\end{minipage}}
		\caption{Average variation of estimated loading space by PCA and RTS methods in 100 replications with growing proportion of outliers when the factor number is $m=1$ (left) and $m=3$ (right).}\label{fig:3}
\end{figure}

We then further compare the RTS and PCA methods by their sensitivity to outliers in this real example.  The sensitivity is evaluated by the variation of estimated loading space $\mathcal{D}(\hat\Lb^{new},\hat\Lb^{old})$, if we randomly select a small proportion of the demeaned log returns in the $105\times 100$ panel and double their values. We repeat the procedure for 100 times to reduce randomness, and report the mean variation in Figure \ref{fig:3} with various contamination levels. It's seen that the estimated loading space can vary a lot with just a small number of outliers. When the contamination level grows, the discrepancy  between the loading spaces $\hat\Lb^{new},\hat\Lb^{old}$ becomes larger, and the phenomenon is more obvious in the case $m=3$ compared with $m=1$. However, the RTS method is less sensitive to outliers than the PCA method in both cases, which is expected as the RTS is more robust.

\section{Discussion}
We proposed a robust two-step estimation procedure for large-dimensional elliptical factor model. In the first step, we estimate the factor loadings by the leading eigenvectors of the spatial Kendall's tau matrix. In the second step, we resort to Ordinary Least Squares (OLS) regression to estimate the factor scores. We  prove the consistency of the proposed estimators for factor loadings, scores and common components.
Numerical studies show that the proposed procedure works comparably with the conventional PCA method when data are from Gaussian distribution while performs much better when data are heavy-tailed, which indicates that the proposed RTS procedure can be used as a safe replacement of the conventional PCA method.

In the future, we aim to propose a robust procedure for more general heavy-tailed data without the constraint of elliptical distribution. In fact, the elliptical assumption exerts  a shape constraint on the distribution of the factors and the idiosyncratic errors, which  may also constrain the real application. For more general case, we may consider the following optimization problem,
\begin{equation}\label{L1}
\bmin_{\Lb, \Fb}\left\{\sum_{i=1}^p\sum_{t=1}^n|y_{it}-\bl_i^\top\bbf_t|\right\},
\end{equation}
which is motivated by  the equivalence of PCA and double least square estimation. We simply replace the quadratic loss function by the absolute loss function in (\ref{L1}) to achieve robustness. We may minimize the absolute loss function alternatively over $\Lb$ and $\Fb$, each time optimizing one argument while keeping the other fixed. The theoretical analysis is more challenging and we leave this as a future work.

\vspace{2em}

\begin{center}
APPENDIX: PROOFS OF MAIN THEOREMS
\end{center}

\begin{appendices}

\section{Proofs of Main Theorems}

We first present three useful lemmas before we give the detailed proofs of main theorems. In the proofs, $c$ denotes some generic finite constant. We denote a random matrix of fixed dimensions as ${\bf o_p(1)}$ or ${\bf O_p(1)}$ when all of its entries are $o_p(1)$ or $O_p(1)$.

\begin{lemma}\label{lemma:lemmanormal}
Assume that $\bg=(g_1,\ldots,g_p)^\top\sim \cN(\zero,\Ib_p)$, then for any $i,j$, we have
	\[
	\mathbb{E}\frac{g_ig_j}{\|\bg\|^2}=0.
	\]
\end{lemma}

\begin{proof}
	Without loss of generality, we take $i=1,j=2$ for example. Define
	\[
	\tilde g_1=\frac{1}{\sqrt{2}}(g_1+g_2),\quad \tilde g_2=\frac{1}{\sqrt{2}}(g_1-g_2),
	\]
	then it's easy to verify that
	\[
	(\tilde g_1,\tilde g_2,g_3,\ldots,g_{p})\sim\mathcal{N}({\bf 0},\Ib_{p}),
	\]
	while by symmetry property,
	\[
		 2\mathbb{E}\frac{g_1g_2}{\|\bg\|^2}=\mathbb{E}\frac{\tilde g_1^2-\tilde g_2^2}{\tilde g_1^2+\tilde g_2^2+\sum_{i=3}^{p}g_i^2}=0.
	\]
\end{proof}

\begin{lemma}\label{lemma:2}
	Assume that $\bg=(g_1,\ldots,g_p)^\top\sim \cN(\zero,\Ib_p)$, then for any $q\times p$ deterministic matrix $\Ab$,
	\[
	 \mathbb{E}\frac{\|\Ab\bg\|^2}{\|\bg\|^2}=\frac{1}{p}\|\Ab\|_F^2.
	\]
\end{lemma}

\begin{proof}
	It's sufficient to prove that the lemma holds with $q=1$. Given a $p$-dimensional deterministic vector $\ba$, we have
	\[
	 \mathbb{E}\frac{(\ba^\top\bg)^2}{\|\bg\|^2}=\sum_{i=1}^{p}\mathbb{E}\frac{(a_ig_i)^2}{\|\bg\|^2}+\sum_{i\ne j}\mathbb{E}\frac{a_ig_i\times a_jg_j}{\|\bg\|^2}.
	\]
	By Lemma \ref{lemma:lemmanormal}, the second term is $0$. For any $i=1,\ldots,p$ we have
	\[
	 \mathbb{E}\frac{(a_ig_i)^2}{\|\bg\|^2}=a_i^2\mathbb{E}\frac{g_i^2}{\|\bg\|^2}=\frac{a_i^2}{p},
	\]
	which concludes the lemma.
\end{proof}

\begin{lemma}\label{LemmaA.1yu}
Under  \textbf{Assumptions A, B, C}, as $\min\{n,p\}\rightarrow\infty$ we have
\[
\left\{
	\begin{matrix}
	\lambda_j(\hat\Kb_y) \asymp m^{-1},&j\le m,\\
	\lambda_j(\hat\Kb_y)=o_p(1),&j>m.
	\end{matrix}
	\right.
\]
\end{lemma}
\begin{proof}
It is adapted from Lemma 3.1 and  Lemma A.1  in \cite{yu2018robust}, so we omit the proof here.
\end{proof}

\section*{Proof of Theorem \ref{theorem:1}}
\begin{proof}
	 Define $\hat\bLambda$ as the diagonal matrix composed of the leading $m$ eigenvalues of $\hat\Kb_y$. Lemma \ref{LemmaA.1yu} implies that $\hat\bLambda$ is asymptotically invertible, $\|\hat\bLambda\|_F=O_p(1)$ and $\|\hat\bLambda^{-1}\|_F=O_p(1)$. Because  $\hat\Lb=\sqrt{p}\hat\bGamma$ and $\hat\bGamma$ is composed of the leading eigenvectors of $\hat\Kb_y$, we have
\[
\hat\Kb_y\hat\Lb=\hat\Lb\hat\bLambda.
\]
Expand $\hat\Kb_y$ by its definition, then
\[
\begin{split}
\hat\Lb\hat\bLambda=&\frac{2}{n(n-1)}\sum_{1\le t<s\le n}\frac{(\by_t-\by_s)(\by_t-\by_s)^\top}{\|\by_t-\by_s\|^2}\hat\Lb\\
=&\frac{2}{n(n-1)}\sum_{1\le t<s\le n}\frac{\big[\Lb(\bbf_t-\bbf_s)+(\bepsilon_t-\bepsilon_s)\big]\big[\Lb(\bbf_t-\bbf_s)+(\bepsilon_t-\bepsilon_s)\big] ^\top}{\|\by_t-\by_s\|^2}\hat\Lb\\
=&\frac{2}{n(n-1)}\sum_{1\le t<s\le n}\frac{\Lb(\bbf_t-\bbf_s)(\bbf_t-\bbf_s)^\top\Lb^\top}{\|\by_t-\by_s\|^2}\hat\Lb+\frac{2}{n(n-1)}\sum_{1\le t<s\le n}\frac{(\bepsilon_t-\bepsilon_s)(\bbf_t-\bbf_s)^\top\Lb^\top}{\|\by_t-\by_s\|^2}\hat\Lb\\
&+\frac{2}{n(n-1)}\sum_{1\le t<s\le n}\frac{\Lb(\bbf_t-\bbf_s)(\bepsilon_t-\bepsilon_s)^\top}{\|\by_t-\by_s\|^2}\hat\Lb+\frac{2}{n(n-1)}\sum_{1\le t<s\le n}\frac{(\bepsilon_t-\bepsilon_s)(\bepsilon_t-\bepsilon_s)^\top}{\|\by_t-\by_s\|^2}\hat\Lb.
\end{split}
\]
For the ease of notations, we denote
\[
\begin{split}
\Mb_1=&\frac{2}{n(n-1)}\sum_{1\le t<s\le n}\frac{(\bbf_t-\bbf_s)(\bbf_t-\bbf_s)^\top}{\|\by_t-\by_s\|^2},\quad\Mb_2=\frac{2}{n(n-1)}\sum_{1\le t<s\le n}\frac{(\bepsilon_t-\bepsilon_s)(\bbf_t-\bbf_s)^\top}{\|\by_t-\by_s\|^2},\\
\Mb_3=&\frac{2}{n(n-1)}\sum_{1\le t<s\le n}\frac{(\bbf_t-\bbf_s)(\bepsilon_t-\bepsilon_s)^\top}{\|\by_t-\by_s\|^2},\quad\Mb_4=\frac{2}{n(n-1)}\sum_{1\le t<s\le n}\frac{(\bepsilon_t-\bepsilon_s)(\bepsilon_t-\bepsilon_s)^\top}{\|\by_t-\by_s\|^2},
\end{split}
\]
and let $\hat\Hb=\Mb_1\Lb^\top\hat\Lb\hat\bLambda^{-1}$, then
\begin{equation}\label{equa1}
\hat\Lb-\Lb\hat\Hb=(\Mb_2\Lb^\top\hat\Lb+\Lb\Mb_3\hat\Lb+\Mb_4\hat\Lb)\hat\bLambda^{-1}.
\end{equation}

\textbf{Lemma S2 } to \textbf{Lemma S4} in the online supplementary  material show that
\[
\|\Mb_2\|_F^2=O_p\bigg(\frac{1}{np}+\frac{1}{p^3}\bigg), \quad 	 \|\Mb_3\|_F^2=O_p\bigg(\frac{1}{np}+\frac{1}{p^3}\bigg),
\]
while
\[
\frac{1}{p}\|\Mb_4\hat\Lb\|_F^2=O_p\bigg(\frac{1}{np}+\frac{1}{p^2}\bigg)+o_p(1)\times \frac{1}{p}\|\hat\Lb-\Lb\hat\Hb\|_F^2.
\]
Therefore, by Cauchy-Schwartz inequality and triangular inequality, it's easy to prove that
\[
\frac{1}{p}\|\hat\Lb-\Lb\hat\Hb\|_F^2=O_p\bigg(\frac{1}{n}+\frac{1}{p^2}\bigg)+o_p(1)\times \frac{1}{p}\|\hat\Lb-\Lb\hat\Hb\|_F^2,
\]
and the convergence rate for $\hat\Lb$ follows directly. To complete the proof, it remains to show $\hat\Hb^\top\Vb\hat\Hb\overset{p}{\rightarrow}\Ib_m$. By Cauchy-Schwartz inequality,
\[
\bigg\|\frac{1}{p}\Lb^\top(\hat\Lb-\Lb\hat\Hb)\bigg\|_F\le \sqrt{\frac{\|\Lb\|_F^2}{p}\frac{\|\hat\Lb-\Lb\hat\Hb\|_F^2}{p}}=o_p(1),\quad \bigg\|\frac{1}{p}\hat\Lb^\top(\hat\Lb-\Lb\hat\Hb)\bigg\|_F=o_p(1).
\]
Note that $p^{-1}\hat\Lb^\top\hat\Lb=\Ib_m$ and $p^{-1}\Lb^\top\Lb\rightarrow \Vb$,  thus it holds that
\[
\frac{1}{p}\Lb^\top\hat\Lb=\Vb\hat\Hb+{\bf o_p(1)},\quad \Ib_m=\frac{1}{p}\hat\Lb^\top\Lb\hat\Hb+{\bf o_p(1)},
\]
which further implies $\hat\Hb^\top\Vb\hat\Hb\overset{p}{\rightarrow}\Ib_m$, and concludes the theorem.
\end{proof}

\section*{Proof of Theorem \ref{theorem:2}}
\begin{proof}
	Because $\hat\Hb^\top\Vb\hat\Hb\overset{p}{\rightarrow}\Ib_m$, we have $\|\hat\Hb\|_F=O_p(1)$ and $\hat\Hb$ is invertible with probability approaching to 1. By our robust estimation procedure,
	\[
	 \hat\bbf_t=\frac{1}{p}\hat\Lb^\top\by_t=\frac{1}{p}\hat\Lb^\top(\Lb\bbf_t+\bepsilon_t)=\frac{1}{p}\hat\Lb^\top\bigg(\hat\Lb\hat\Hb^{-1}-(\hat\Lb\hat\Hb^{-1}-\Lb)\bigg)\bbf_t+\frac{1}{p}\hat\Lb^\top\bepsilon_t.
	\]
	Note that $p^{-1}\hat\Lb^\top\hat\Lb=\Ib_m$, then
	\[
	 \hat\bbf_t-\hat\Hb^{-1}\bbf_t=\frac{1}{p}\hat\Lb^\top(\hat\Lb-\Lb\hat\Hb)\hat\Hb^{-1}\bbf_t+\frac{1}{p}(\hat\Lb-\Lb\hat\Hb)^\top\bepsilon_t+\frac{1}{p}\hat\Hb^\top\Lb^\top\bepsilon_t.
	\]
	
	\textbf{Lemma S5} and \textbf{Lemma S6} in our online supplementary material show that
	\[
	 \bigg\|\frac{1}{p}\hat\Lb^\top(\hat\Lb-\Lb\hat\Hb)\bigg\|_F^2=O_p\bigg(\frac{1}{n^2}+\frac{1}{p^2}\bigg),\quad \bigg\|\frac{1}{p}(\hat\Lb-\Lb\hat\Hb)^\top\bepsilon_t\bigg\|_F^2=O_p\bigg(\frac{1}{n^2}+\frac{1}{p^2}\bigg).
	\]
Meanwhile, for any $t\le n$, by \textbf{Assumption A} and Lemma \ref{lemma:2} we have
\[
\|\bbf_t\|^2=\bigg(\frac{\zeta_t}{\sqrt{p}}\bigg)^2\frac{\|\sqrt{p}(\Ib_m,{\bf 0})\bg_t\|^2}{\|\bg_t\|^2}=O_p(1).
\]
Similarly, it's not hard to prove that for any $t\le n$,
\[
\bigg\|\frac{1}{p}\Lb^\top\bepsilon_t\bigg\|_F^2=\frac{1}{p^2}\bigg(\frac{\zeta_t}{\sqrt{p}}\bigg)^2\frac{\|\sqrt{p}\Lb^\top({\bf 0},\Ab)\bg_t\|^2}{\|\bg_t\|^2}=O_p\bigg(\frac{\|\Lb^\top\Ab\|_F^2}{p^2}\bigg)=O_p\Big(\frac{1}{p}\Big).
\]
Hence, by Cauchy-Schwartz inequality and  triangular inequality, we have for any $t\le n$,
\[
\|\hat\Hb\hat\bbf_t-\bbf_t\|^2=O_p(n^{-2}+p^{-1}),
\]
which concludes the theorem.

	\end{proof}

\section*{Proof of Theorem \ref{theorem:3}}
\begin{proof}
	By Theorem \ref{theorem:1} and Theorem \ref{theorem:2}, we already have
	\[
	 \frac{1}{p}\|\hat\Lb-\Lb\hat\Hb\|_F^2=O_p\Big(\frac{1}{n}+\frac{1}{p^2}\Big),\text{ and } \|\hat\Hb\hat\bbf_t-\bbf_t\|^2=O_p\Big(\frac{1}{p}+\frac{1}{n^2}\Big) \text{ for any $t\le n$}.
	\]
	 Hence, by triangular inequality and Cauchy-Schwartz inequality, we have for any $t\le n$,
	\[
	\begin{split}
	 \frac{1}{p}\|\hat\Lb\hat\bbf_t^\top-\Lb\bbf_t^\top\|^2=&	 \frac{1}{p}\|\hat\Lb\hat\Hb^{-1}\hat\Hb\hat\bbf_t-\hat\Lb\hat\Hb^{-1}\bbf_t+\hat\Lb\hat\Hb^{-1}\bbf_t-\Lb\bbf_t\|^2\\
	\le &\Big(\frac{2}{p}\|\hat\Lb\hat\Hb^{-1}\|_F^2\Big)\Big(\|\hat\Hb\hat\bbf_t-\bbf_t^\top\|^2\Big)+
\Big(\frac{2}{p}\|\hat\Lb\hat\Hb^{-1}-\Lb\|_F^2\Big)\|\bbf_t\|^2\\
	=&O_p\Big(\frac{1}{{n}}+\frac{1}{p }\Big),
	\end{split}
	\]
which concludes the theorem.
\end{proof}

\end{appendices}







\bibliographystyle{model2-names}

\bibliography{Ref}

\begin{thebibliography}{26}
\expandafter\ifx\csname natexlab\endcsname\relax\def\natexlab#1{#1}\fi
\providecommand{\url}[1]{\texttt{#1}}
\providecommand{\href}[2]{#2}
\providecommand{\path}[1]{#1}
\providecommand{\DOIprefix}{doi:}
\providecommand{\ArXivprefix}{arXiv:}
\providecommand{\URLprefix}{URL: }
\providecommand{\Pubmedprefix}{pmid:}
\providecommand{\doi}[1]{\href{http://dx.doi.org/#1}{\path{#1}}}
\providecommand{\Pubmed}[1]{\href{pmid:#1}{\path{#1}}}
\providecommand{\bibinfo}[2]{#2}
\ifx\xfnm\relax \def\xfnm[#1]{\unskip,\space#1}\fi
\bibitem[{Ahn and Horenstein(2013)}]{Ahn2013Eigenvalue}
\bibinfo{author}{Ahn, S.C.}, \bibinfo{author}{Horenstein, A.R.},
  \bibinfo{year}{2013}.
\newblock \bibinfo{title}{Eigenvalue ratio test for the number of factors}.
\newblock \bibinfo{journal}{Econometrica} \bibinfo{volume}{81},
  \bibinfo{pages}{1203--1227}.
\bibitem[{Bai(2003)}]{Bai2003Inferential}
\bibinfo{author}{Bai, J.}, \bibinfo{year}{2003}.
\newblock \bibinfo{title}{Inferential theory for factor models of large
  dimensions}.
\newblock \bibinfo{journal}{Econometrica} \bibinfo{volume}{71},
  \bibinfo{pages}{135--171}.
\bibitem[{Bai and Li(2012)}]{Bai2012Statistical}
\bibinfo{author}{Bai, J.}, \bibinfo{author}{Li, K.}, \bibinfo{year}{2012}.
\newblock \bibinfo{title}{Statistical analysis of factor models of high
  dimension}.
\newblock \bibinfo{journal}{The Annals of Statistics} \bibinfo{volume}{40},
  \bibinfo{pages}{436--465}.
\bibitem[{Bai and Li(2014)}]{Bai2014Theory}
\bibinfo{author}{Bai, J.}, \bibinfo{author}{Li, K.}, \bibinfo{year}{2014}.
\newblock \bibinfo{title}{Theory and methods of panel data models with
  interactive effects}.
\newblock \bibinfo{journal}{The Annals of Statistics} \bibinfo{volume}{42},
  \bibinfo{pages}{142--170}.
\bibitem[{Bai and Li(2016)}]{Bai2016Maximum}
\bibinfo{author}{Bai, J.}, \bibinfo{author}{Li, K.}, \bibinfo{year}{2016}.
\newblock \bibinfo{title}{Maximum likelihood estimation and inference for
  approximate factor models of high dimension}.
\newblock \bibinfo{journal}{Review of Economics and Statistics}
  \bibinfo{volume}{98}, \bibinfo{pages}{298--309}.
\bibitem[{Bai and Ng(2002)}]{Bai2002Determining}
\bibinfo{author}{Bai, J.}, \bibinfo{author}{Ng, S.}, \bibinfo{year}{2002}.
\newblock \bibinfo{title}{Determining the number of factors in approximate
  factor models}.
\newblock \bibinfo{journal}{Econometrica} \bibinfo{volume}{70},
  \bibinfo{pages}{191--221}.
\bibitem[{Chamberlain and Rothschild(1983)}]{Chamberlain1983Arbitrage}
\bibinfo{author}{Chamberlain, G.}, \bibinfo{author}{Rothschild, M.},
  \bibinfo{year}{1983}.
\newblock \bibinfo{title}{Arbitrage, factor structure, and mean-variance
  analysis on large asset markets}.
\newblock \bibinfo{journal}{Econometrica} \bibinfo{volume}{51},
  \bibinfo{pages}{1281--1304}.
\bibitem[{Choi and Marden(1998)}]{Choi1998A}
\bibinfo{author}{Choi, K.}, \bibinfo{author}{Marden, J.}, \bibinfo{year}{1998}.
\newblock \bibinfo{title}{A multivariate version of kendall's $\tau$}.
\newblock \bibinfo{journal}{Journal of Nonparametric Statistics}
  \bibinfo{volume}{9}, \bibinfo{pages}{261--293}.
\bibitem[{Cont(2001)}]{Cont2001Empirical}
\bibinfo{author}{Cont, R.}, \bibinfo{year}{2001}.
\newblock \bibinfo{title}{Empirical properties of asset returns: stylized facts
  and statistical issues}.
\newblock \bibinfo{journal}{Quantitative Finance} \bibinfo{volume}{1},
  \bibinfo{pages}{223--236}.
\bibitem[{Croux et~al.(2002)Croux, Ollila and Oja}]{croux2002sign}
\bibinfo{author}{Croux, C.}, \bibinfo{author}{Ollila, E.},
  \bibinfo{author}{Oja, H.}, \bibinfo{year}{2002}.
\newblock \bibinfo{title}{Sign and rank covariance matrices: statistical
  properties and application to principal components analysis}, in:
  \bibinfo{booktitle}{Statistical data analysis based on the L1-norm and
  related methods}, pp. \bibinfo{pages}{257--269}.
\bibitem[{Fama(1963)}]{Fama1963Mandelbrot}
\bibinfo{author}{Fama, E.F.}, \bibinfo{year}{1963}.
\newblock \bibinfo{title}{Mandelbrot and the stable paretian hypothesis}.
\newblock \bibinfo{journal}{Journal of Business} \bibinfo{volume}{36},
  \bibinfo{pages}{420--429}.
\bibitem[{Fan et~al.(2013)Fan, Liao and Mincheva}]{fan2013large}
\bibinfo{author}{Fan, J.}, \bibinfo{author}{Liao, Y.},
  \bibinfo{author}{Mincheva, M.}, \bibinfo{year}{2013}.
\newblock \bibinfo{title}{Large covariance estimation by thresholding principal
  orthogonal complements}.
\newblock \bibinfo{journal}{Journal of the Royal Statistical Society: Series B
  (Statistical Methodology)} \bibinfo{volume}{75}, \bibinfo{pages}{603--680}.
\bibitem[{Fan et~al.(2018)Fan, Liu and Wang}]{fan2018}
\bibinfo{author}{Fan, J.}, \bibinfo{author}{Liu, H.}, \bibinfo{author}{Wang,
  W.}, \bibinfo{year}{2018}.
\newblock \bibinfo{title}{Large covariance estimation through elliptical factor
  models}.
\newblock \bibinfo{journal}{The Annals of Statistics} \bibinfo{volume}{46},
  \bibinfo{pages}{1383--1414}.
\bibitem[{Han and Liu(2012)}]{han2012semiparametric}
\bibinfo{author}{Han, F.}, \bibinfo{author}{Liu, H.}, \bibinfo{year}{2012}.
\newblock \bibinfo{title}{Semiparametric principal component analysis}, in:
  \bibinfo{booktitle}{Advances in Neural Information Processing Systems}, pp.
  \bibinfo{pages}{171--179}.
\bibitem[{Han and Liu(2014)}]{Han2014Scale}
\bibinfo{author}{Han, F.}, \bibinfo{author}{Liu, H.}, \bibinfo{year}{2014}.
\newblock \bibinfo{title}{Scale-invariant sparse {PCA} on high-dimensional
  meta-elliptical data}.
\newblock \bibinfo{journal}{Journal of the American Statistical Association}
  \bibinfo{volume}{109}, \bibinfo{pages}{275--287}.
\bibitem[{Han and Liu(2018)}]{han2018eca}
\bibinfo{author}{Han, F.}, \bibinfo{author}{Liu, H.}, \bibinfo{year}{2018}.
\newblock \bibinfo{title}{\text{ECA}: High-dimensional elliptical component
  analysis in non-gaussian distributions}.
\newblock \bibinfo{journal}{Journal of the American Statistical Association}
  \bibinfo{volume}{113}, \bibinfo{pages}{252--268}.
\bibitem[{Jing et~al.(2012)Jing, Kong and Liu}]{jing2012modeling}
\bibinfo{author}{Jing, B.Y.}, \bibinfo{author}{Kong, X.B.},
  \bibinfo{author}{Liu, Z.}, \bibinfo{year}{2012}.
\newblock \bibinfo{title}{Modeling high-frequency financial data by pure jump
  processes}.
\newblock \bibinfo{journal}{The Annals of Statistics} \bibinfo{volume}{40},
  \bibinfo{pages}{759--784}.
\bibitem[{Kong et~al.(2015)Kong, Liu and Jing}]{kong2015testing}
\bibinfo{author}{Kong, X.B.}, \bibinfo{author}{Liu, Z.}, \bibinfo{author}{Jing,
  B.Y.}, \bibinfo{year}{2015}.
\newblock \bibinfo{title}{Testing for pure-jump processes for high-frequency
  data}.
\newblock \bibinfo{journal}{The Annals of Statistics} \bibinfo{volume}{43},
  \bibinfo{pages}{847--877}.
\bibitem[{Marden(1999)}]{Marden1999Some}
\bibinfo{author}{Marden, J.I.}, \bibinfo{year}{1999}.
\newblock \bibinfo{title}{Some robust estimates of principal components}.
\newblock \bibinfo{journal}{Statistics \& Probability Letters}
  \bibinfo{volume}{43}, \bibinfo{pages}{349--359}.
\bibitem[{Onatski(2009)}]{onatski2009testing}
\bibinfo{author}{Onatski, A.}, \bibinfo{year}{2009}.
\newblock \bibinfo{title}{Testing hypotheses about the number of factors in
  large factor models}.
\newblock \bibinfo{journal}{Econometrica} \bibinfo{volume}{77},
  \bibinfo{pages}{1447--1479}.
\bibitem[{Stock and Watson(2002a)}]{stock2002forecast}
\bibinfo{author}{Stock, J.H.}, \bibinfo{author}{Watson, M.W.},
  \bibinfo{year}{2002}a.
\newblock \bibinfo{title}{Forecasting using principal components from a large
  number of predictors}.
\newblock \bibinfo{journal}{Journal of the American Statistical Association}
  \bibinfo{volume}{97}, \bibinfo{pages}{1167--1179}.
\bibitem[{Stock and Watson(2002b)}]{Stock2002Macroeconomic}
\bibinfo{author}{Stock, J.H.}, \bibinfo{author}{Watson, M.W.},
  \bibinfo{year}{2002}b.
\newblock \bibinfo{title}{Macroeconomic forecasting using diffusion indexes}.
\newblock \bibinfo{journal}{Journal of Business \& Economic Statistics}
  \bibinfo{volume}{20}, \bibinfo{pages}{147--162}.
\bibitem[{Trapani(2018)}]{Trapani2018A}
\bibinfo{author}{Trapani, L.}, \bibinfo{year}{2018}.
\newblock \bibinfo{title}{A randomised sequential procedure to determine the
  number of factors}.
\newblock \bibinfo{journal}{Journal of the American Statistical Association}
  \bibinfo{volume}{113}, \bibinfo{pages}{1341--1349}.
\bibitem[{Visuri et~al.(2000)Visuri, Koivunen and Oja}]{Visuri2000Sign}
\bibinfo{author}{Visuri, S.}, \bibinfo{author}{Koivunen, V.},
  \bibinfo{author}{Oja, H.}, \bibinfo{year}{2000}.
\newblock \bibinfo{title}{Sign and rank covariance matrices}.
\newblock \bibinfo{journal}{Journal of Statistical Planning \& Inference}
  \bibinfo{volume}{91}, \bibinfo{pages}{557--575}.
\bibitem[{Xia et~al.(2017)Xia, Liang and Wu}]{Xia2017Transformed}
\bibinfo{author}{Xia, Q.}, \bibinfo{author}{Liang, R.}, \bibinfo{author}{Wu,
  J.}, \bibinfo{year}{2017}.
\newblock \bibinfo{title}{Transformed contribution ratio test for the number of
  factors in static approximate factor models}.
\newblock \bibinfo{journal}{Computational Statistics \& Data Analysis}
  \bibinfo{volume}{112}, \bibinfo{pages}{235--241}.
\bibitem[{Yu et~al.(2019)Yu, He and Zhang}]{yu2018robust}
\bibinfo{author}{Yu, L.}, \bibinfo{author}{He, Y.}, \bibinfo{author}{Zhang,
  X.}, \bibinfo{year}{2019}.
\newblock \bibinfo{title}{Robust factor number specification for
  large-dimensional elliptical factor model}.
\newblock \bibinfo{journal}{Journal of Multivariate analysis}
  \bibinfo{volume}{174}, \bibinfo{pages}{104543}.

\end{thebibliography}

\newpage

\title{\LARGE Supplementary Material for ``Large-dimensional Factor Analysis without Moment Constraints" }

\vspace{2em}

In the supplementary material, we give some useful lemmas and the corresponding detailed proofs. \textbf{Lemma S1} provides some technical error bounds which facilitate  presentation of the following proofs.    \textbf{Lemma S2 } to \textbf{Lemma S6 } are critical to the proof of the main theorems and are mentioned in the main paper. We let $\Xb$  and $\Zb$ denote  some generic random matrices in our proof and can be nonidentical in different lemmas. For two series $a_n$ and $b_n$, $a_n\lesssim b_n$ means that $a_n=O(b_n)$ as $n\rightarrow\infty$ (or $a_n=O_p(b_n)$ for random series).

	\vspace{1em}
\noindent\textbf{Lemma S1.}\label{lems1}
		Under \textbf{Assumptions A, B, C}, as $p\rightarrow\infty$ we have
		\[
		 \|\Ib_m-\Lb^\top\bSigma_y^{-1}\Lb\|_F^2=O(p^{-2}),\quad \|\Lb^\top\bSigma_y^{-2}\Lb\|_F^2=O(p^{-2}),
		\]
		where $\Lb$ is the loading matrix and $\bSigma_y=\Lb\Lb^\top+\bSigma_{\epsilon}$ is the scatter matrix of $\by_t$.
	\begin{proof}
		Because $\bSigma_{\epsilon}$ is positive definite, we have $\bSigma_y$ is always positive definite and invertible. Denote the Singular Value Decomposition (SVD) of $\Lb$ as $\Lb=\Pb\Db\Qb$. $\Pb$ is $p\times p$ orthogonal matrix. The upper $m\times m$ sub-matrix of $\Db$ ($p\times m$) is  diagonal and denoted as $\Db_1$, while the left are 0. $\Qb$ is $m\times m$ orthogonal matrix. Hence, the diagonal entries of $\Db_1$ are of order $\sqrt{p}$ and
\[
\begin{split}
\Ib_m-\Lb^\top\bSigma_y^{-1}\Lb=&\Qb^\top\Qb-\Qb^\top\Db^\top\Pb^\top(\Pb\Db\Db^\top\Pb^\top+\bSigma_{\epsilon})^{-1}\Pb\Db\Qb\\
=&\Qb^\top\bigg\{\Ib_m-\Db^\top(\Db\Db^\top+\Pb^\top\bSigma_{\epsilon}\Pb)^{-1}\Db\bigg\}\Qb.
\end{split}
\]
Partitioning $\Pb^\top\bSigma_{\epsilon}\Pb$ into $\left(\begin{matrix}
\Ab_1&\Ab_2\\
\Ab_3&\Ab_4
\end{matrix}\right)$ where $\Ab_1$ is $m\times m$ submatrix, then $\Ab_4$ is invertible, $\|\Ab_i\|\le c$ for $i=1,2,3,4$ and
\[
\begin{split}
\Db^\top(\Db\Db^\top+\Pb^\top\bSigma_{\epsilon}\Pb)^{-1}\Db=&\Db_1(\Db_1\Db_1^\top+\Ab_1-\Ab_2\Ab_4^{-1}\Ab_3)^{-1}\Db_1\\
=&\bigg(\Ib_m+\Db_1^{-1}(\Ab_1-\Ab_2\Ab_4^{-1}\Ab_3)\Db_1^{-1}\bigg)^{-1}.
\end{split}
\]
By Cauchy-Schwartz inequality, it's easy to prove that $\|\Db_1^{-1}(\Ab_1-\Ab_2\Ab_4^{-1}\Ab_3)\Db_1^{-1}\|=O(p^{-1})$,  then the first part of this lemma holds.

For the second part, denote the eigenvalue decomposition of $\bSigma_y$ as $\bSigma_y=\bGamma_1\bLambda_1\bGamma_1^\top+\bGamma_2\bLambda_2\bGamma_2^\top$, where $\bGamma_1$ is composed of the leading $m$ eigenvectors. By Weyl's theorem, it's easy to prove that $\lambda_j(\bSigma_y)\asymp p$ for $j\le m$ and $\lambda_j(\bSigma_y)\le c$ for $j>m$. Hence, by definition,
\[
(\Lb\Lb^\top+\bSigma_{\epsilon})\bGamma_1=\bGamma_1\bLambda_1\Rightarrow \bGamma_1=\Lb\Lb^\top\bGamma_1\bLambda_1^{-1}+\bSigma_{\epsilon}\bGamma_1\bLambda_1^{-1}.
\]
Let $\Hb=\sqrt{p}\Lb^\top\bGamma_1\bLambda_1^{-1}$, then $\|\Hb\|_F=O(1)$ and
\[
\bGamma_1-\frac{1}{\sqrt{p}}\Lb\Hb=\bSigma_{\epsilon}\bGamma_1\bLambda_1^{-1}.
\]
Therefore,
\[
\frac{1}{\sqrt{p}}\Lb^\top\bGamma_1=\frac{1}{p}\Lb^\top\Lb\Hb+\frac{1}{\sqrt{p}}\Lb^\top\bSigma_{\epsilon}\bGamma_1\bLambda_1^{-1}={\bf O(1)}.
\]
Further,
\[
\frac{1}{\sqrt{p}}\Lb^\top\bGamma_2=-(\Hb^\top)^{-1}\bigg(\bGamma_1-\frac{1}{\sqrt{p}}\Lb\Hb-\bGamma_1\bigg)^\top\bGamma_2=-(\Hb^\top)^{-1}\bigg(\bGamma_1-\frac{1}{\sqrt{p}}\Lb\Hb\bigg)^\top\bGamma_2.
\]
Hence, by Cauchy-Schwartz inequality,
\[
\bigg\|\frac{1}{\sqrt{p}}\Lb^\top\bGamma_2\bigg\|_F^2\lesssim\bigg\|\bGamma_1-\frac{1}{\sqrt{p}}\Lb\Hb\bigg\|_F^2\le \|\bSigma_{\epsilon}\|^2\|\bGamma_1\|^2\|\bLambda_1^{-1}\|_F^2=O(p^{-2}).
\]
Now, we can calculate that
\[
\begin{split}
\|\Lb^\top\bSigma_y^{-2}\Lb\|_F^2=&\|\Lb^\top(\bGamma_1\bLambda_1^{-2}\bGamma_1^\top+\bGamma_2\bLambda_2^{-2}\bGamma_2^\top)\Lb\|_F^2\\
\lesssim&p^2\times \bigg\|\frac{1}{\sqrt{p}}\Lb^\top\bGamma_1\bigg\|_F^4\|\bLambda_1^{-2}\|^2+p^2\times\bigg\|\frac{1}{\sqrt{p}}\Lb^\top\bGamma_2\bigg\|_F^4\|\bLambda_2^{-2}\|^2=O(p^{-2}),
\end{split}
\]
which concludes the second part of the lemma.
	\end{proof}

\vspace{1em}
\noindent\textbf{Lemma S2.}
Under \textbf{Assumptions A, B, C}, as $\min\{n,p\}\rightarrow\infty$,  we have
		\[
		 \|\Mb_2\|_F^2=O_p\bigg(\frac{1}{np}+\frac{1}{p^3}\bigg),\quad \|\Mb_3\|_F^2=O_p\bigg(\frac{1}{np}+\frac{1}{p^3}\bigg),
		\]
		where $\Mb_2$ and $\Mb_3$ are defined in the proof of Theorem \ref{theorem:1}.
\begin{proof}

	Note that $\Mb_3$ is $\Mb_2$'s transpose so we only show the result with $\Mb_2$.  Assume $n$ is even and $\bar n=n/2$, otherwise we can delete the last observation. Given a permutation of $\{1,\ldots,n\}$, denoted as $\sigma$,  let $\bbf_t^\sigma$, $\bepsilon_t^\sigma$ and $\by_t^\sigma$ be the rearranged factors, errors and observations, further define
	\[
	\Mb_2^\sigma=\frac{1}{\bar n}\sum_{s=1}^{\bar n}\frac{(\bepsilon_{2s-1}^\sigma-\bepsilon_{2s}^\sigma)(\bbf_{2s-1}^\sigma-\bbf_{2s}^\sigma)^\top}{\|\by_{2s-1}^\sigma-\by_{2s}^\sigma\|^2}.
	\]
	Denote $\mathcal{S}_n$ as the set containing all the permutations of $\{1,\ldots,n\}$, then it's not hard to prove that
	\[
	\sum_{\sigma\in \mathcal{S}_n}\frac{n}{2}\Mb_2^\sigma=n\times (n-2)!\times \frac{n(n-1)}{2}\Mb_2.
	\]
	That is,
	\[
	\Mb_2=\frac{1}{n!}\sum_{\sigma\in \mathcal{S}_n}\Mb_2^\sigma\quad \Rightarrow\quad \mathbb{E}\|\Mb_2\|_F\le \frac{1}{n!}\sum_{\sigma\in \mathcal{S}_n}\mathbb{E}\|\Mb_2^\sigma\|_F=\mathbb{E}\|\Mb_2^\sigma\|_F\le \sqrt{\mathbb{E}\|\Mb_2^{\sigma}\|_F^2}\text{ for any $\sigma$}.
	\]
	Now take $\sigma$ as given, i.e., $\sigma=\{1,\ldots,n\}$ which is the original order.  By the property of elliptical distribution, for any $s=1,\ldots,\bar n$,
	\[
\left(\begin{matrix}
\bbf_{2s-1}-\bbf_{2s}\\
\bepsilon_{2s-1}-\bepsilon_{2s}\\
\end{matrix}
\right)\overset{d}{=}\xi_1\left(\begin{matrix}
\Ib_m&{\bf 0}\\
{\bf 0}&\Ab\\
\end{matrix}
\right)\frac{\bg}{\|\bg\|},
	\]
	where $\xi_1$ is determined by $\xi$. $\bg\sim\mathcal{N}_{m+p}({\bf 0},\Ib)$ and $\bg$ is independent of $\xi_1$. $\Ab\Ab^\top=\bSigma_{\epsilon}$. Hence,
	\[
	 \Xb_s:=\frac{(\bepsilon_{2s-1}^\sigma-\bepsilon_{2s}^\sigma)(\bbf_{2s-1}^\sigma-\bbf_{2s}^\sigma)^\top}{\|\by_{2s-1}^\sigma-\by_{2s}^\sigma\|^2}\overset{d}{=}\frac{\Ab\bg_{2}\bg_{1}^\top}{\|\Lb\bg_{1}+\Ab\bg_{2}\|^2},
	\]
	where $\bg_1$ is composed of the first $m$ entries of $\bg$ and $\bg_2$ is composed of the left ones.  Note that $\Xb_s$ and $\Xb_t$ are independently and identically distributed when $s\ne t$,  so
	\begin{equation}\label{equs1}
	 \mathbb{E}\|\Mb_2^{\sigma}\|_F^2=\mathbb{E}\bigg\|\frac{1}{\bar n}\sum_{s=1}^{\bar n}\Xb_s\bigg\|_F^2=\frac{1}{\bar n}\mathbb{E}\|\Xb_1\|_F^2+\frac{\bar n(\bar n-1)}{\bar n^2}\|\mathbb{E}\Xb_1\|_F^2.
	\end{equation}
	
	We first focus on the matrix $\mathbb{E}\Xb_1$. Define
	\[
	\left(\begin{matrix}
	\bu_1\\
	\bu_2\\
	\end{matrix}\right)=\left(\begin{matrix}
	\Ib_m&-\Lb^\top\bSigma_y^{-1}\\
	{\bf 0}&\Ib_p\\
	\end{matrix}\right)\left(\begin{matrix}
	\Ib_m&{\bf 0}\\
	\Lb&\Ab\\
	\end{matrix}\right)\left(\begin{matrix}
	\bg_1\\
	\bg_2\\
	\end{matrix}\right)
	\sim\mathcal{N}\bigg({\bf 0},\left(\begin{matrix}
\bSigma_{\bu_1}&{\bf 0}\\
	{\bf 0}&\bSigma_y\\
	\end{matrix}\right)\bigg),
	\]
	where $\bSigma_{\bu_1}=	 \Ib_m-\Lb^\top\bSigma_y^{-1}\Lb$, then $\bu_1$ and $\bu_2$ are independent and
	\[
	\begin{split}
	\left(\begin{matrix}
\bg_1\\
\bg_2\\
\end{matrix}\right)=&\left(\begin{matrix}
\Ib_m&{\bf 0}\\
-\Ab^{-1}\Lb&\Ab^{-1}\\
\end{matrix}\right)\left(\begin{matrix}
\Ib_m&\Lb^\top\bSigma_y^{-1}\\
{\bf 0}&\Ib_p\\
\end{matrix}\right)\left(\begin{matrix}
\bu_1\\
\bu_2\\
\end{matrix}\right)
= \left(\begin{matrix}
\Ib_m&\Lb^\top\bSigma_y^{-1}\\
-\Ab^{-1}\Lb&\Ab^\top\bSigma_y^{-1}\\
\end{matrix}\right)\left(\begin{matrix}
\bu_1\\
\bu_2\\
\end{matrix}\right).
	\end{split}
	\]
As a result,
	\begin{equation}\label{equs2}
	 \frac{\Ab\bg_{2}\bg_1^\top}{\|\Lb\bg_{1}+\Ab\bg_{2}\|^2}=\frac{(-\Lb\bu_1+\Ab\Ab^\top\bSigma_y^{-1}\bu_2)(\bu_1+\Lb^\top\bSigma_y^{-1}\bu_2)^\top}{\|\bu_2\|^2}.
	\end{equation}
	Because $\bu_1$ and $\bu_2$ are zero-mean independent Gaussian vectors, we have
	\[
	 \mathbb{E}\frac{\bu_2\bu_1^\top}{\|\bu_2\|^2}={\bf 0},\quad \mathbb{E}\|\bu_2\|^{-2}\le \frac{1}{\lambda_p(\bSigma_y)}\mathbb{E}\frac{1}{\chi_p^2}\asymp p^{-1},
	\]
where $\chi_p$ is a Chi-square random variable with degree $p$. Hence,
\begin{equation}\label{equs3}
	 \mathbb{E}\Xb_1=-\mathbb{E}\|\bu_2\|^{-2}\Lb\bSigma_{\bu_1}+\Ab\Ab^\top\bSigma_y^{-1}\mathbb{E}\frac{\bu_2\bu_2^\top}{\|\bu_2\|^2}\bSigma_y^{-1}\Lb.
\end{equation}
By \textbf{Lemma S1}, $\|\bSigma_{\bu_1}\|_F^2=O(p^{-2})$, then we have
\begin{equation}\label{equs4}
\bigg\|-\mathbb{E}\|\bu_2\|^{-2}\Lb\bSigma_{\bu_1}\bigg\|_F^2\le \|\Lb\|_F^2\|\bSigma_{\bu_1}\|_F^2\Big(\mathbb{E}\|\bu_2\|^{-2}\Big)^2=O(p^{-3}).
\end{equation}
For the second term, denote the spectral decomposition of $\bSigma_y$ as $\bGamma_y\bLambda_y\bGamma_y^\top$, then
\[
\bigg\|\bSigma_y^{-1}\mathbb{E}\frac{\bu_2\bu_2^\top}{\|\bu_2\|^2}\bSigma_y^{-1}\Lb\bigg\|_F^2=\bigg\|\bLambda_y^{-1}\mathbb{E}\frac{\bGamma_y^\top\bu_2\bu_2^\top\bGamma_y}{\|\bGamma_y^\top\bu_2\|^2}\bLambda_y^{-1}\bGamma_y^\top\Lb\bigg\|_F^2.
\]
Denote
\[
\tilde\bLambda:=\mathbb{E}\frac{\bGamma_y^\top\bu_2\bu_2^\top\bGamma_y}{\|\bGamma_y^\top\bu_2\|^2},
\]
then $\tilde\bLambda$ is diagonal with positive diagonal entries and
\[
\tilde\bLambda\le \bLambda_y^{1/2}\mathbb{E}\frac{\bLambda_y^{-1/2}\bGamma_y^\top\bu_2\bu_2^\top\bGamma_y\bLambda_y^{-1/2}}{\lambda_p(\bSigma_y)\times\|\bLambda_y^{-1/2}\bGamma_y^\top\bu_2\|^2}\bLambda_y^{1/2}=\frac{1}{p\lambda_p(\bSigma_y)}\bLambda_y.
\]
Consequently,
\begin{equation}\label{equs5}
\bigg\|\bSigma_y^{-1}\mathbb{E}\frac{\bu_2\bu_2^\top}{\|\bu_2\|^2}\bSigma_y^{-1}\Lb\bigg\|_F^2\le \bigg\|\frac{1}{p\lambda_p(\bSigma_y)}\bLambda_y^{-1}\bGamma_y^\top\Lb\bigg\|_F^2\lesssim p^{-2}\|\Lb^\top\bSigma_y^{-2}\Lb\|_F=O(p^{-3}). \text{ (by \textbf{Lemma S1})}
\end{equation}
Combing equations (\ref{equs3}), (\ref{equs4}) and (\ref{equs5}) we have $\|\mathbb{E}\Xb_1\|_F^2=O(p^{-3})$.

Now we move to the calculation of $\mathbb{E}\|\Xb_1\|_F^2$.  Note that
\[
\mathbb{E}\|\bu_1\|^2\lesssim p^{-1},\quad \mathbb{E}\|\bu_1\|^4\le 2^m\times 3\mathbb{E}\|\bu_1\|^2\lesssim p^{-1},\quad \mathbb{E}\|\bu_2\|^{-4}\le \frac{1}{\lambda_p^2(\bSigma_y)}\mathbb{E}\chi_p^{-2}\asymp p^{-2},
\]
then  by Cauchy-Schwartz inequality,
\[
\begin{split}
\mathbb{E}\|\Xb_1\|_F^2\lesssim& \mathbb{E}\frac{(\|\Lb\bu_1\|^2+\|\Ab\Ab^\top\bSigma_y^{-1}\bu_2\|^2)(\|\bu_1\|^2+\|\Lb^\top\bSigma_y^{-1}\bu_2\|^2)}{\|\bu_2\|^4}\\
\lesssim&\mathbb{E}\bigg(\frac{\|\Lb\|_F^2\|\bu_1\|^4}{\|\bu_2\|^4}+\frac{\|\bSigma_y^{-1}\|^2\|\bu_1\|^2+\|\Lb\|_F^2\|\Lb^\top\bSigma_y^{-1}\|^2\|\bu_1\|^2}{\|\bu_2\|^2}+\|\bSigma_y^{-1}\|^2\|\Lb^\top\bSigma_y^{-1}\|^2\bigg)\\
=&O(p^{-1}).
\end{split}
\]

As a result,
\[
\mathbb{E}\|\Mb_2^{\sigma}\|_F\lesssim\sqrt{(np)^{-1}+p^{-3}}\Rightarrow \|\Mb_2\|_F^2=O_p\bigg(\frac{1}{np}+\frac{1}{p^3}\bigg).
\]
which concludes the lemma.
\end{proof}

\vspace{1em}
\noindent\textbf{Lemma S3.}
Under \textbf{Assumptions A, B, C},  we have $\|\Mb_1\|_F^2=O_p(p^{-2})$ and $\|\hat\Hb\|_F^2=O_p(1)$.
\begin{proof}
	Similar to  the proof of equations (\ref{equs1}) and (\ref{equs2}) in \textbf{Lemma S2}, it's easy to verify that
\[
\mathbb{E}\|\Mb_1\|_F\le\sqrt{\frac{1}{\bar n}\mathbb{E}\|\Xb\|_F^2+\frac{\bar n(\bar n-1)}{\bar n^2}\|\mathbb{E}\Xb\|_F^2},\text{ where } \Xb\overset{d}{=}\frac{(\bu_1+\Lb^\top\bSigma_y^{-1}\bu_2)(\bu_1+\Lb^\top\bSigma_y^{-1}\bu_2)}{\|\bu_2\|^2},
\]
and  $\bu_1$ and $\bu_2$ are the same as in the proof of \textbf{Lemma S2}. Therefore,
\[
\|\mathbb{E}\Xb\|_F^2\lesssim \|\bSigma_{\bu_1}\|_F^2(\mathbb{E}\|\bu_2\|^{-2})^2+ \bigg\|\Lb^\top\bSigma_y^{-1}\mathbb{E}\frac{\bu_2\bu_2^\top}{\|\bu_2\|^2}\bSigma_y^{-1}\Lb\bigg\|_F^2=O(p^{-2}).
\]
On the other hand,
\[
\mathbb{E}\|\Xb\|_F^2\lesssim \mathbb{E}\|\bu_1\|^4\mathbb{E}\|\bu_2\|^{-4}+\mathbb{E}\|\bu_1\|^2\|\Lb^\top\bSigma_y^{-1}\|^2\mathbb{E}\|\bu_2\|^{-2}+\|\Lb^\top\bSigma_y^{-1}\|_F^4=O(p^{-2}).
\]
As a result, $\|\Mb_1\|_F^2=O_p(p^{-2})$ and by the definition of $\hat\Hb$ we have $\|\hat\Hb\|_F^2\le\|\Mb_1\|_F^2\|\Lb\|_F^2\|\hat\Lb\|_F^2\|\hat\bLambda^{-1}\|_F^2=O_p(1)$, where $\hat\bLambda$ is composed of the leading $m$ eigenvalues of $\hat\Kb_y$.
\end{proof}

\vspace{1em}
\noindent\textbf{Lemma S4.}
Under \textbf{Assumptions A, B, C},  we have
	\[
	 \frac{1}{p}\|\Mb_4\hat\Lb\|_F^2=O_p\bigg(\frac{1}{n}+\frac{1}{p^2}\bigg)+o_p(1)\times \frac{1}{p}\|\hat\Lb-\Lb\hat\Hb\|.
	\]
\begin{proof}
	By the decomposition that  $\hat\Lb=\hat\Lb-\Lb\hat\Hb+\Lb\hat\Hb$, we have
	\begin{equation}\label{equs6}
	 \frac{1}{p}\|\Mb_4\hat\Lb\|_F^2\lesssim  \frac{1}{p}\|\Mb_4\Lb\|_F^2\|\hat\Hb\|_F^2+\|\Mb_4\|_F^2\times \frac{1}{p}\|\hat\Lb-\Lb\hat\Hb\|_F^2.
	\end{equation}
	
	We start with the matrix $\Mb_4\Lb$. 	Similar to  the proof of equations (\ref{equs1}) and (\ref{equs2}) in \textbf{Lemma S2}, we have
	\[
	\mathbb{E}\|\Mb_4\Lb\|_F\le\sqrt{\frac{1}{\bar n}\mathbb{E}\|\Xb\|_F^2+\frac{\bar n(\bar n-1)}{\bar n^2}\|\mathbb{E}\Xb\|_F^2},\text{ where } \Xb\overset{d}{=}\frac{(-\Lb\bu_1+\Ab\Ab^\top\bSigma_y^{-1}\bu_2)(-\Lb\bu_1+\Ab\Ab^\top\bSigma_y^{-1}\bu_2)^\top\Lb}{\|\bu_2\|^2},
\]
and  $\bu_1$ and $\bu_2$ are the same as in the proof of  \textbf{Lemma S2}. Therefore,
\[
\|\mathbb{E}\Xb\|_F^2\lesssim \|\Lb\|_F^2\|\bSigma_{\bu_1}\|_F^2\|\Lb^\top\Lb\|_F^2(\mathbb{E}\|\bu_2\|^{-2})^2+ \bigg\|\bSigma_y^{-1}\mathbb{E}\frac{\bu_2\bu_2^\top}{\|\bu_2\|^2}\bSigma_y^{-1}\bigg\|^2\|\Lb\|_F^2=O(p^{-1}).
\]
On the other hand,
\[
\begin{split}
\mathbb{E}\|\Xb\|_F^2\lesssim& \|\Lb\|_F^2\mathbb{E}\|\bu_1\|^4\|\Lb^\top\Lb\|_F^2\mathbb{E}\|\bu_2\|^{-4}+\|\bSigma_y^{-1}\|^2\mathbb{E}\|\bu_1\|^2\|\Lb^\top\Lb\|_F^2\mathbb{E}\|\bu_2\|^{-2}\\
&+\|\Lb\|_F^2\mathbb{E}\|\bu_1\|^2\|\bSigma_y^{-1}\Ab\Ab^\top\Lb\|_F^2\mathbb{E}\|\bu_2\|^{-2}+\|\bSigma_y^{-1}\|^2\|\bSigma_y^{-1}\Ab\Ab^\top\Lb\|_F^2=O(p).
\end{split}
\]
As a result,
\[
\|\Mb_4\Lb\|_F^2=O_p\bigg(\frac{p}{n}+\frac{1}{p}\bigg).
\]

Based on equation (\ref{equs6}), it remains to show that $\|\Mb_4\|_F^2=o_p(1)$. Similarly to  the proof of equations (\ref{equs1}) and (\ref{equs2}) again, we have
	\[
\mathbb{E}\|\Mb_4\|_F\le\sqrt{\frac{1}{\bar n}\mathbb{E}\|\Zb\|_F^2+\frac{\bar n(\bar n-1)}{\bar n^2}\|\mathbb{E}\Zb\|_F^2},\text{ where } \Zb\overset{d}{=}\frac{(-\Lb\bu_1+\Ab\Ab^\top\bSigma_y^{-1}\bu_2)(\Lb\bu_1+\Ab\Ab^\top\bSigma_y^{-1}\bu_2)^\top}{\|\bu_2\|^2}.
\]
On one hand,
\[
\|\mathbb{E}\Zb\|_F^2\lesssim \|\Lb\|_F^2\|\bSigma_{\bu_1}\|_F^2\|\Lb\|_F^2(\mathbb{E}\|\bu_2\|^{-2})^2+\bigg\|\bSigma_y^{-1}\mathbb{E}\frac{\bu_2\bu_2^\top}{\|\bu_2\|^2}\bSigma_y^{-1}\bigg\|_F^2=O(p^{-1}).
\]
On the other hand,
\[
\mathbb{E}\|\Zb\|_F^2\lesssim \|\Lb\|_F^4\mathbb{E}\|\bu_1\|^4\mathbb{E}\|\bu_2\|^{-4}+\|\bSigma_y^{-1}\|^2\|\Lb\|_F^2\mathbb{E}\|\bu_1\|^2+(\|\bSigma_y^{-1}\|^2)^2=O(1).
\]
Hence, $\|\Mb_4\|_F^2=O_p(n^{-1}+p^{-1})=o_p(1)$ and the lemma holds.
 \end{proof}

\vspace{1em}
\noindent\textbf{Lemma S5.}
Under \textbf{Assumptions A, B, C},  we have
	\[
	 \bigg\|\frac{1}{p}\hat\Lb^\top(\hat\Lb-\Lb\hat\Hb)\bigg\|_F^2=O_p\bigg(n^{-2}+p^{-2}\bigg).
	\]
\begin{proof}
	Note that
	\[
	 \frac{1}{p}\hat\Lb^\top(\hat\Lb-\Lb\hat\Hb)=\frac{1}{p}(\hat\Lb-\Lb\hat\Hb)^\top(\hat\Lb-\Lb\hat\Hb)+\frac{1}{p}\hat\Hb^\top\Lb^\top(\hat\Lb-\Lb\hat\Hb),
	\]
	while $p^{-1}\|\hat\Lb-\Lb\hat\Hb\|_F^2=O_p(n^{-1}+p^{-2})$ and $\|\hat\Hb\|_F^2=O_p(1)$. Hence it suffices to show that
	\[
	 \bigg\|\frac{1}{p}\Lb^\top(\hat\Lb-\Lb\hat\Hb)\bigg\|_F^2=O_p(n^{-2}+p^{-2}).
	\]
	By equation (\ref{equa1}),
	\begin{equation}\label{equs7}
	 \frac{1}{p}\Lb^\top(\hat\Lb-\Lb\hat\Hb)=\bigg(\frac{1}{p}\Lb^\top\Mb_2\Lb^\top\hat\Lb+\frac{1}{p}\Lb^\top\Lb\Mb_3\hat\Lb+\frac{1}{p}\Lb^\top\Mb_4\hat\Lb\bigg)\hat\bLambda^{-1}.
	\end{equation}
	
	We will calculate the three error terms separately. Firstly, similarly to  the proof of equations (\ref{equs1}) and (\ref{equs2}) in \textbf{Lemma S2}, we have
\[
\mathbb{E}\|\Lb^\top\Mb_2\|_F\lesssim\sqrt{n^{-1}\mathbb{E}\|\Xb\|_F^2+\|\mathbb{E}\Xb\|_F^2},\text{ where } \Xb\overset{d}{=}\frac{\Lb^\top(-\Lb\bu_1+\Ab\Ab^\top\bSigma_y^{-1}\bu_2)(\bu_1+\Lb^\top\bSigma_y^{-1}\bu_2)^\top}{\|\bu_2\|^2}.
\]
Then, it's not hard to show that
\[
\begin{split}
\|\mathbb{E}\Xb\|_F^2\lesssim&\|\Lb^\top\Lb\|_F^2\|\bSigma_{\bu_1}\|_F^2(\mathbb{E}\|\bu_2\|^{-2})^2+\|\Lb\|_F^2\bigg\|\bSigma_y^{-1}\mathbb{E}\frac{\bu_2\bu_2^\top}{\|\bu_2\|^2}\bSigma_y^{-1}\Lb\bigg\|_F^2=O(p^{-2}),\\
\mathbb{E}\|\Xb\|_F^2\lesssim&\|\Lb^\top\Lb\|_F^2\mathbb{E}\|\bu_1\|^4\mathbb{E}\|\bu_2\|^{-4}+\|\Lb\|_F^2\mathbb{E}\|\bu_1\|^2\mathbb{E}\|\bu_2\|^{-2}+\|\Lb^\top\Lb\|_F^2\|\Lb^\top\bSigma_y^{-1}\|_F^2\mathbb{E}\|\bu_1\|^2\mathbb{E}\|\bu_2\|^{-2}\\
&+\|\Lb\|_F^2\text{tr}\bigg(\Lb^\top\bSigma_y^{-1}\mathbb{E}\frac{\bu_2\bu_2^\top}{\|\bu_2\|^2}\bSigma_y^{-1}\Lb\bigg)=O(p^{-1}).
\end{split}
\]
As a result,
\begin{equation}\label{equs8}
\|\Lb^\top\Mb_2\|_F^2=O_p\bigg(\frac{1}{np}+\frac{1}{p^2}\bigg).
\end{equation}
Similarly, $\|\Mb_3\Lb\|_F^2=O_p(n^{-1}p^{-1}+p^{-2})$ and
\begin{equation}\label{equs9}
\|\Mb_3\hat\Lb\|_F^2\lesssim \|\Mb_3\Lb\|_F^2+\|\Mb_3\|_F^2\|\hat\Lb-\Lb\hat\Hb\|_F^2=O_p\bigg(\frac{1}{n^2}+\frac{1}{np}+\frac{1}{p^2}\bigg).
\end{equation}
For the last term, we can verify that
\[
\mathbb{E}\|\Lb^\top\Mb_4\Lb\|_F\lesssim\sqrt{n^{-1}\mathbb{E}\|\Zb\|_F^2+\|\mathbb{E}\Zb\|_F^2},\text{ where } \Zb\overset{d}{=}\frac{\Lb^\top(\Lb\bu_1+\bSigma_{\epsilon}\bSigma_y^{-1}\bu_2)(\Lb\bu_1+\bSigma_{\epsilon}\bSigma_y^{-1}\bu_2)^\top\Lb^\top}{\|\bu_2\|^2}.
\]
Some simple calculations lead to
\[
\begin{split}
\|\mathbb{E}\Zb\|_F^2\lesssim&\|\Lb^\top\Lb\|_F^4\|\bSigma_{\bu_1}\|_F^2\Big(\mathbb{E}\|\bu_2\|^{-2}\Big)^2+\|\Lb\|_F^4\bigg\|\bSigma_y^{-1}\mathbb{E}\frac{\bu_2\bu_2^\top}{\|\bu_2\|^2}\bSigma_y^{-1}\bigg\|^2=O(1),\\
\mathbb{E}\|\Zb\|_F^2\lesssim&\|\Lb^\top\Lb\|_F^4\mathbb{E}\|\bu_1\|^4\mathbb{E}\|\bu_2\|^{-4}+\|\Lb\|_F^2\|\Lb^\top\Lb\|_F^2\mathbb{E}\|\bu_1\|^2\mathbb{E}\|\bu_2\|^{-2}\\
&+\|\Lb\|_F^2\text{tr}\bigg(\Lb^\top\bSigma_{\epsilon}\bSigma_y^{-1}\mathbb{E}\frac{\bu_2\bu_2^\top}{\|\bu_2\|^2}\bSigma_y^{-1}\bSigma_{\epsilon}\Lb\bigg)=O(p).
\end{split}
\]
Hence, $\|\Lb^\top\Mb_4\Lb\|_F^2=O_p(1+p/n)$ and
\begin{equation}\label{equs10}
\bigg\|\frac{1}{p}\Lb^\top\Mb_4\hat\Lb\bigg\|_F^2\lesssim \frac{1}{p^2}\bigg(\|\Lb^\top\Mb_4\Lb\|_F^2+\|\Lb^\top\Mb_4\|_F^2\|\hat\Lb-\Lb\hat\Hb\|_F^2\bigg)=O_p\bigg(\frac{1}{n^2}+\frac{1}{np}+\frac{1}{p^2}\bigg).
\end{equation}
Combine equations (\ref{equs7}), (\ref{equs8}), (\ref{equs9}) and (\ref{equs10}) we have
\[
	 \bigg\|\frac{1}{p}\Lb^\top(\hat\Lb-\Lb\hat\Hb)\bigg\|_F^2=O_p\bigg(\frac{1}{n^2}+\frac{1}{np}+\frac{1}{p^2}\bigg)=O_p(n^{-2}+p^{-2}),
\]
which concludes the lemma.
\end{proof}

\vspace{1em}
\noindent\textbf{Lemma S6.}
Under \textbf{Assumptions A, B, C}, we have for any $t\le n$,
\[
\bigg\|\frac{1}{p}(\hat\Lb-\Lb\hat\Hb)^\top\bepsilon_t\bigg\|_F^2=O_p(n^{-2}+p^{-2}).
\]
\begin{proof}
	Since $\bepsilon_1,\ldots,\bepsilon_n$ are i.i.d.,  we assume $t=1$. Apply equation (\ref{equa1}) again, we have
	\begin{equation}\label{equs11}
	 \frac{1}{p}(\hat\Lb-\Lb\hat\Hb)^\top\bepsilon_1=\frac{1}{p}\bigg(\hat\Lb^\top\Lb\Mb_2^\top\bepsilon_1+\hat\Lb^\top\Mb_3^\top\Lb^\top\bepsilon_1+\hat\Lb^\top\Mb_4^\top\bepsilon_1\bigg)\hat\bLambda^{-1}.
	\end{equation}
	
	We calculate the three error terms separately. Firstly by the definition of $\Mb_2$, we have
	\[
	 \Mb_2^\top\bepsilon_1=\frac{2}{n(n-1)}\sum_{s=2}^{n}\frac{(\bbf_1-\bbf_s)(\bepsilon_1-\bepsilon_s)^\top}{\|\by_1-\by_s\|^2}\bepsilon_1+\frac{2}{n(n-1)}\sum_{2\le s<s^\prime\le n}\frac{(\bbf_s-\bbf_{s^\prime})(\bepsilon_s-\bepsilon_{s^\prime})^\top}{\|\by_s-\by_{s^\prime}\|^2}\bepsilon_1.
	\]
	For the first term, by Cauchy-Schwartz inequality we have
	\[
	 \bigg\|\frac{2}{n(n-1)}\sum_{s=2}^{n}\frac{(\bbf_1-\bbf_s)(\bepsilon_1-\bepsilon_s)^\top}{\|\by_1-\by_s\|^2}\bepsilon_1\bigg\|_F^2\lesssim n^{-3}\sum_{s=2}^n\bigg\|\frac{(\bbf_1-\bbf_s)(\bepsilon_1-\bepsilon_s)^\top}{\|\by_1-\by_s\|^2}\bigg\|_F^2\|\bepsilon_1\|^2.
	\]
	By the proof of \textbf{Lemma S2}, for any $s=2,\ldots,n$,
	\[
	 \mathbb{E}\bigg\|\frac{(\bbf_1-\bbf_s)(\bepsilon_1-\bepsilon_s)^\top}{\|\by_1-\by_s\|^2}\bigg\|_F^2=O(p^{-1}).
	\]
	while by Lemma \ref{lemma:2} we have for any $t\le T$,
	\[
	\|\bepsilon_t\|^2=\bigg(\frac{\zeta_t}{\sqrt{p}}\bigg)^2\frac{\|\sqrt{p}({\bf 0},\Ab)\bg_t\|^2}{\|\bg_t\|^2}=O_p(\|\Ab\|_F^2)=O_p(p).
	\]
     Hence,
	\[
	 \bigg\|\frac{2}{n(n-1)}\sum_{s=2}^{n}\frac{(\bbf_1-\bbf_s)(\bepsilon_1-\bepsilon_s)^\top}{\|\by_1-\by_s\|^2}\bepsilon_1\bigg\|_F^2=O_p(n^{-2}).
	\]
	For the second term,
	\[
	\frac{2}{n(n-1)}\sum_{2\le s<s^\prime\le n}\frac{(\bbf_s-\bbf_{s^\prime})(\bepsilon_s-\bepsilon_{s^\prime})^\top}{\|\by_s-\by_{s^\prime}\|^2}\bepsilon_1=\frac{\zeta_1}{\sqrt{p}}\times\frac{2}{n(n-1)}\sum_{2\le s<s^\prime\le n}\frac{(\bbf_s-\bbf_{s^\prime})(\bepsilon_s-\bepsilon_{s^\prime})^\top}{\|\by_s-\by_{s^\prime}\|^2}\frac{\sqrt{p}({\bf 0}, \Ab)\bg}{\|\bg\|^2},
	\]
	where $\bg\sim\mathcal{N}(0,\Ib_{m+p})$ and is independent of $\{\bbf_s,\bepsilon_s\}_{2\le s\le n}$. By Lemma \ref{lemma:2},
	\[
	\mathbb{E}\bigg\|\frac{2}{n(n-1)}\sum_{2\le s<s^\prime\le n}\frac{(\bbf_s-\bbf_{s^\prime})(\bepsilon_s-\bepsilon_{s^\prime})^\top}{\|\by_s-\by_{s^\prime}\|^2}\frac{\sqrt{p}({\bf 0}, \Ab)\bg}{\|\bg\|^2}\bigg\|^2\le	 \mathbb{E}\bigg\|\frac{2}{n(n-1)}\sum_{2\le s<s^\prime\le n}\frac{(\bbf_s-\bbf_{s^\prime})(\bepsilon_s-\bepsilon_{s^\prime})^\top}{\|\by_s-\by_{s^\prime}\|^2}\bigg\|_F^2\|\Ab\|.
	\]
	Removing the first observations will not change the large-sample property of $\Mb_2$, then
	\[
	\bigg\|\frac{2}{n(n-1)}\sum_{2\le s<s^\prime\le n}\frac{(\bbf_s-\bbf_{s^\prime})(\bepsilon_s-\bepsilon_{s^\prime})^\top}{\|\by_s-\by_{s^\prime}\|^2}\bepsilon_1\bigg\|_F^2=O_p\bigg(\|\Mb_2\|_F^2\bigg)=O_p\bigg(\frac{1}{np}+\frac{1}{p^3}\bigg).
	\]
	As a result,
	\begin{equation}\label{equs12}
	\|\Mb_2^\top\bepsilon_1\|^2\le O_p\bigg(\frac{1}{n^2}+\frac{1}{np}+\frac{1}{p^3}\bigg)
	\end{equation}
	
	Secondly,  it's easy that
	\begin{equation}\label{equs13}
	 \bigg\|\frac{1}{p}\hat\Lb^\top\Mb_3^\top\Lb^\top\bepsilon_1\bigg\|^2\le p^{-2}\|\hat\Lb\|_F^2\|\Mb_3\|_F^2\|\Lb^\top\bepsilon_1\|^2=O_p\bigg(\frac{1}{np}+\frac{1}{p^3}\bigg).
	\end{equation}
	
	Thirdly,
	\[
	 \frac{1}{p}\hat\Lb^\top\Mb_4^\top\bepsilon_1=\frac{1}{p}(\hat\Lb-\Lb\hat\Hb)^\top\Mb_4^\top\bepsilon_1+\frac{1}{p}\hat\Hb^\top\Lb^\top\Mb_4^\top\bepsilon_1.
	\]
	For the first term, note that in the proof of \textbf{Lemma S4} we get $\|\Mb_4\|_F^2=O_p(n^{-1}+p^{-1})$, then
	\[
	 \bigg\|\frac{1}{p}(\hat\Lb-\Lb\hat\Hb)^\top\Mb_4^\top\bepsilon_1\bigg\|_F^2\le \frac{1}{p}\|\hat\Lb-\Lb\hat\Hb\|_F^2\|\Mb_4\|^2\frac{1}{p}\|\bepsilon_1\|^2\le O_p\bigg(\frac{1}{n^2}+\frac{1}{p^2}\bigg).
	\]
For the second term, by the definition of $\Mb_4$ we have
\[
\Lb^\top\Mb_4^\top\bepsilon_1=\frac{2}{n(n-1)}\sum_{s=2}^{n}\frac{\Lb^\top(\bepsilon_1-\bepsilon_s)(\bepsilon_1-\bepsilon_s)^\top}{\|\by_1-\by_s\|^2}\bepsilon_1+\frac{2}{n(n-1)}\sum_{2\le s<s^\prime\le n}\frac{\Lb^\top(\bepsilon_s-\bepsilon_{s^\prime})(\bepsilon_s-\bepsilon_{s^\prime})^\top}{\|\by_s-\by_{s^\prime}\|^2}\bepsilon_1.
\]
By a similar technique as how we control $\|\Mb_2^\top\bepsilon_1\|^2$, it's not hard to show that
\[
\bigg\|\frac{2}{n(n-1)}\sum_{s=2}^{n}\frac{\Lb^\top(\bepsilon_1-\bepsilon_s)(\bepsilon_1-\bepsilon_s)^\top}{\|\by_1-\by_s\|^2}\bepsilon_1\bigg\|^2\le O_p\bigg(\frac{p^2}{n^2}\bigg),
\]
while
\[
\bigg\|\frac{2}{n(n-1)}\sum_{2\le s<s^\prime\le n}\frac{\Lb^\top(\bepsilon_s-\bepsilon_{s^\prime})(\bepsilon_s-\bepsilon_{s^\prime})^\top}{\|\by_s-\by_{s^\prime}\|^2}\bepsilon_1\bigg\|^2\le O_p\bigg(\|\Lb^\top\Mb_4\|_F^2\bigg)=O_p\bigg(\frac{p}{n}+\frac{1}{p}\bigg).
\]
Therefore,
\begin{equation}\label{equs14}
\bigg\|\frac{1}{p}\hat\Lb^\top\Mb_4^\top\bepsilon_1\bigg\|^2=O_p\bigg(\frac{1}{n^2}+\frac{1}{p^2}+\frac{1}{np}\bigg)\le O_p\bigg(\frac{1}{n^2}+\frac{1}{p^2}\bigg).
\end{equation}
Combine equations (\ref{equs11}), (\ref{equs12}), (\ref{equs13}) and (\ref{equs14}) then the lemma holds.
\end{proof}

\end{document}